\documentclass{article}

\usepackage[final, nonatbib]{nips_2017}

\usepackage[utf8]{inputenc} 
\usepackage[T1]{fontenc}    
\usepackage{hyperref}       
\usepackage{url}            
\usepackage{epsf}
\usepackage{booktabs}       
\usepackage{amsfonts}       
\usepackage{nicefrac}       
\usepackage{microtype}      
\usepackage{color}
\usepackage{graphicx}
\usepackage{subfigure}

\usepackage[linesnumbered,ruled]{algorithm2e}
\RequirePackage{amsthm,amsmath,amsfonts,amssymb}

\usepackage{chngcntr}
\usepackage{float}

\newcommand{\CC}{\mathcal{C}}
\newcommand{\G}{\mathcal{G}}
\newcommand{\HH}{\mathcal{H}}

\newcommand{\Score}{\textrm{Score}}
\newcommand{\M}{\textrm{maximize}}

\newcommand{\R}{\mathbb{R}}
\newcommand{\PP}{\mathbb{P}}
\newcommand{\union}{\cup}

\newcommand{\I}{\mathcal{I}}

\newcommand{\cc}{\checkmark}

\DeclareMathOperator{\parents}{Pa}
\DeclareMathOperator{\descendants}{De}
\DeclareMathOperator{\nondescendants}{Nd}
\DeclareMathOperator{\ancestors}{An}
\DeclareMathOperator{\children}{Ch}

\DeclareMathOperator*{\argmax}{\arg\!\max}
\DeclareMathOperator*{\pa}{Pa}

 \newcommand\independent{\protect\mathpalette{\protect\independenT}{\perp}}
    \def\independenT#1#2{\mathrel{\rlap{$#1#2$}\mkern2mu{#1#2}}}

\SetKwInput{KwInput}{Input}
\SetKwInput{KwOutput}{Output}

\newtheorem{theorem}{Theorem}[section]
\newtheorem{lemma}[theorem]{Lemma}

\theoremstyle{definition}
\newtheorem{definition}[theorem]{Definition}

\newenvironment{ideas}{\par\color{blue}}{\par}

\definecolor{ist green}{RGB}{57,74,19}
\definecolor{kth blue}{RGB}{26,84,166}

\title{Permutation-based Causal Inference Algorithms \\ with Interventions}

\author{
  Yuhao Wang \\
  Laboratory for Information and Decision Systems\\
 and Institute for Data, Systems and Society\\
  Massachusetts Institute of Technology\\
  Cambridge, MA 02139 \\
  \texttt{yuhaow@mit.edu} \\
  \And
  Liam Solus \\
  Department of Mathematics\\
  KTH Royal Institute of Technology\\
  Stockholm, Sweden \\
  \texttt{solus@kth.se} \\
    \And
     Karren Dai Yang \\
Institute for Data, Systems and Society\\
and Broad Institute of MIT and Harvard\\
  Massachusetts Institute of Technology\\
  Cambridge, MA 02139 \\
  \texttt{karren@mit.edu} \\
  \And
  Caroline Uhler \\
  Laboratory for Information and Decision Systems\\
 and Institute for Data, Systems and Society\\
  Massachusetts Institute of Technology\\
  Cambridge, MA 02139 \\
  \texttt{cuhler@mit.edu }
}

\begin{document}

\maketitle

\begin{abstract}
Learning directed acyclic graphs using both observational and interventional data is now a fundamentally important problem due to recent technological developments in genomics that generate such single-cell gene expression data at a very large scale.  
In order to utilize this data for learning gene regulatory networks, efficient and reliable causal inference algorithms are needed that can make use of both observational and interventional data.  
In this paper, we present two algorithms of this type and prove that both are consistent under the faithfulness assumption.  
These algorithms are interventional adaptations of the Greedy SP algorithm and are the first algorithms using both observational and interventional data with consistency guarantees.  
Moreover, these algorithms have the advantage that they are nonparametric, which makes them useful also for analyzing non-Gaussian data.
In this paper, we present these two algorithms and their consistency guarantees, and we analyze their performance on simulated data, protein signaling data, and single-cell gene expression data.

\end{abstract}

\section{Introduction}
\label{sec: introduction}

Discovering causal relations is a fundamental problem across a wide variety of disciplines including computational biology, epidemiology, sociology, and economics \cite{Friedman_2000,Pearl_2000,Robins_2000,Spirtes_2001}.  
DAG models can be used to encode causal relations in terms of a directed acyclic graph (DAG) $\G$, where each node is associated to a random variable and the arrows represent their causal influences on one another.  
The non-arrows of $\G$ encode a collection of conditional independence (CI) relations through the so-called \emph{Markov properties}. 
While DAG models are extraordinarily popular within the aforementioned research fields, it is in general a difficult task to recover the underlying DAG $\G$ from samples from the joint distribution on the nodes.  
In fact, since different DAGs can encode the same set of CI relations, from observational data alone the underlying DAG $\G$ is in general only identifiable up to \emph{Markov equivalence}, and \emph{interventional} data is needed to identify the complete DAG.

In recent years, the new \emph{drop-seq} technology has allowed obtaining high-resolution observational single-cell gene expression data at a very large scale~\cite{drop-seq}. In addition, earlier this year  this technology was combined with the CRISPR/Cas9 system into \emph{perturb-seq}, a technology that allows obtaining high-throughput interventional gene expression data~\cite{perturb-seq}. An imminent question now is how to make use of a combination of observational and interventional data (of the order of 100,000 cells / samples on 20,000 genes / variables) in the causal discovery process. 
Therefore, the development of efficient and consistent algorithms using both observational and interventional data that are implementable within genomics is now a crucial goal.  
This is the purpose of the present paper.  

The remainder of this paper is structured as follows: In Section~\ref{sec: related work} we discuss related work. Then in Section~\ref{sec: preliminaries and related work}, we recall fundamental facts about DAG models and causal inference that we will use in the coming sections.  
In Section~\ref{sec: two permutation-based algorithms}, we present the two algorithms and discuss their consistency guarantees.  
In Section~\ref{sec: evaluation}, we analyze the performance of the two algorithms on both simulated and real datasets.  
We end with a short discussion in Section~\ref{sec: discussion}.

\section{Related Work}
\label{sec: related work}

Causal inference algorithms based on observational data can be classified into three categories: constraint-based, score-based, and hybrid methods. 
\emph{Constraint-based methods}, such as the \emph{PC algorithm}~\cite{Spirtes_2001}, treat causal inference as a constraint satisfaction problem and rely on CI tests to recover the model via its Markov properties. \emph{Score-based methods}, on the other hand, assign a score function such as the Bayesian Information Criterion (BIC) to each DAG and optimize the score via greedy approaches. An example is  the prominent \emph{Greedy Equivalence Search (GES)}~\cite{M97}. \emph{Hybrid methods}
either alternate between score-based and constraint-based updates, as in \emph{Max-Min Hill-Climbing}~\cite{Tsamardinos}, or use score functions based on CI tests, as in the recently introduced \emph{Greedy SP} algorithm~\cite{SWUM17}.

Based on the growing need for efficient and consistent algorithms that accommodate observational and interventional data \cite{perturb-seq}, it is natural to consider extensions of the previously described algorithms that can accommodate interventional data.  
Such options have been considered in~\cite{HB12}, in which the authors propose GIES, an extension of GES that accounts for interventional data.
This algorithm can be viewed as a greedy approach to \emph{$\ell_0$-penalized maximum likelihood estimation} with interventional data, an otherwise computationally infeasible score-based approach. Hence GIES is a parametric approach (relying on Gaussianity) and while it has been applied to real data~\cite{HB12, HB15, Meinshausen}, we will demonstrate via an example in Section~\ref{sec: preliminaries and related work} that it is in general not consistent.
In this paper, we assume causal sufficiency, i.e., that there are no latent confounders in the data-generating DAG. In addition, we assume that the interventional targets are known. 
Methods such as ACI~\cite{MCM16}, HEJ~\cite{HEJ14}, COmbINE~\cite{TT15} and ICP~\cite{Meinshausen} allow for latent confounders with possibly unknown interventional targets. In addition, other methods have been developed specifically for the analysis of gene expression data~\cite{RJN13}. 
A comparison of the method presented here and some of these methods in the context of gene expression data is given in the Supplementary Material. 

The main purpose of this paper is to provide the first algorithms (apart from enumerating all DAGs) for causal inference based on observational and interventional data with consistency guarantees. These algorithms are adaptations of the Greedy SP algorithm~\cite{SWUM17}. As compared to GIES,  another advantage of these algorithms is that they are nonparametric and hence do not assume Gaussianity, a feature that is crucial for applications to gene expression data which is inherently non-Gaussian.
\section{Preliminaries}
\label{sec: preliminaries and related work}


{\bf DAG models.}
Given a DAG $\G = ([p], A)$ with node set $[p]:=\{1,\ldots,p\}$ and a collection of arrows $A$, we associate the nodes of $\G$ to a random vector $(X_1,\ldots,X_p)$ with joint probability distribution $\PP$.  
For a subset of nodes $S\subset[p]$, we let $\parents_\G(S)$, $\ancestors_\G(S)$, $\children_\G(S)$, $\descendants_\G(S)$, and $\nondescendants_\G(S)$, denote the \emph{parents, ancestors, children, descendants}, and \emph{nondescendants} of $S$ in $\G$.  
Here, we use the typical graph theoretical definitions of these terms as given in \cite{Lauritzen_1996}.  
By the Markov property, the collection of non-arrows of $\G$ encode a set of CI relations
$
X_i \independent X_{\nondescendants(i)\backslash\parents(i)} \,\mid\, X_{\parents(i)}.
$
A distribution $\PP$ is said to satisfy the \emph{Markov assumption} (a.k.a.~be \emph{Markov}) with respect to $\mathcal{G}$ if it entails these CI relations.
A fundamental result about DAG models is that the complete set of CI relations implied by the Markov assumption for $\G$ is given by the \emph{$d$-separation} relations in $\G$~\cite[Section~3.2.2]{Lauritzen_1996}; i.e., $\PP$ satisfies the Markov assumption with respect to $\G$ if and only if  $X_A\independent X_B \,\mid\, X_C$ in $\PP$ whenever $A$ and $B$ are $d$-separated in $\G$ given $C$. 
The \emph{faithfulness assumption} is the assertion that the only CI relations entailed by $\PP$ are those implied by $d$-separation in $\G$.   

Two DAGs $\G$ and $\HH$ with the same set of $d$-separation statements are called \emph{Markov equivalent}, and the complete set of DAGs that are Markov equivalent to $\G$ is called its \emph{Markov equivalence class} (MEC), denoted $[\G]$.  
The MEC of $\G$ is represented combinatorially by a partially directed graph $\widehat{\G} := ([p],D,E)$, called its \emph{CP-DAG} or \emph{essential graph} \cite{AMP97}.  
The arrows $D$ are precisely those arrows in $\G$ that have the same orientation in all members of $[\G]$, and the edges $E$ represent those arrows that change direction between distinct members of the MEC.  
In \cite{C95}, the authors give a transformational characterization of the members of $[\G]$.  
An arrow $i\rightarrow j$ in $\G$ is called a \emph{covered arrow} if $\parents_\G(j) = \parents_\G(i)\cup\{i\}$.  
Two DAGs $\G$ and $\HH$ are Markov equivalent if and only if there exists a sequence of covered arrow reversals transforming $\G$ into $\HH$ \cite{C95}.  
This transformational characterization plays a fundamental role in GES~\cite{M97}, GIES~\cite{HB12}, Greedy SP~\cite{SWUM17}, as well as the algorithms we introduce in this paper.  

{\bf Learning from Interventions.}
In this paper, we consider multiple \emph{interventions}.  
Given an ordered list of subsets of $[p]$ denoted by $\I := \{I_1,I_2,\ldots, I_K\}$, for each $I_j$ we generate an \emph{interventional distribution}, denoted $\PP^j$, by forcing the random variables $X_i$ for $i\in I_j$ to the value of some independent random variables. 
We assume throughout that $I_j=\emptyset$ for some $j$, i.e., that we have access to a combination of observational and interventional data. 
If $\PP$ is Markov with respect to $\G = ([p],A)$, then the \emph{intervention DAG of $I_j$} is the subDAG $\G^j := ([p],A^j)$ where $A^j = \{(i,j)\in A : j\notin I_j\}$; i.e., $\G^j$ is given by removing the incoming arrows to all intervened nodes in $\G$.  
Notice that $\PP^j$ is always Markov with respect to $\G^j$.  
This fact allows us to naturally extend the notions of Markov equivalence and essential graphs to the interventional setting, as described in \cite{HB12}.  
Two DAGs $\G$ and $\HH$ are \emph{$\I$-Markov equivalent} for the collection of interventions $\I$ if they have the same skeleton and the same set of immoralities, and if $\G^j$ and $\HH^j$ have the same skeleton for all $j = 1,\ldots,K$ \cite[Theorem 10]{HB12}. Hence, any two $\I$-Markov equivalent DAGs lie in the same MEC. The \emph{$\I$-Markov equivalence class} ($\I$-MEC) of $\G$ is denoted $[\G]_\I$.  
The \emph{$\I$-essential graph} of $\G$ is the partially directed graph $\widehat{\G}_\I := \left([p], \cup_{j = 1}^KD^j,\cup_{j = 1}^KE^j\right)$, where $\widehat{\G}^j =([p],D^j,E^j)$.  
The arrows of $\widehat{\G}_\I$ are called \emph{$\I$-essential arrows} of $\G$.

%
{\bf Greedy Interventional Equivalence Search (GIES).} GIES is a three-phase score-based algorithm:  
In the \emph{forward phase}, GIES initializes with an empty $\I$-essential graph $\widehat{\G}_0$.  
Then it sequentially steps from one $\I$-essential graph $\widehat{\G}_i$ to a larger one $\widehat{\G}_{i+1}$ given by adding a single arrow to $\widehat{\G}_i$.  
In the \emph{backward phase}, it steps from one essential graph $\widehat{\G}_i$ to a smaller one $\widehat{\G}_{i+1}$ containing precisely one less arrow than $\widehat{\G}_i$.  
In the \emph{turning phase}, the algorithm reverses the direction of arrows.  
It first considers reversals of non-$\I$-essential arrows and then the reversal of $\I$-essential arrows, allowing it to move between $\I$-MECs.
At each step in all phases the maximal scoring candidate is chosen, and the phase is only terminated when no higher-scoring $\I$-essential graph exists.  
GIES repeatedly executes the forward, backward, and turning phases, in that order, until no higher-scoring $\I$-essential graph can be found.  
It is amenable to any score that is constant on an $I$-MEC, such as the BIC. 

	\begin{figure}
	\centering
	\includegraphics[height = 1.5in]{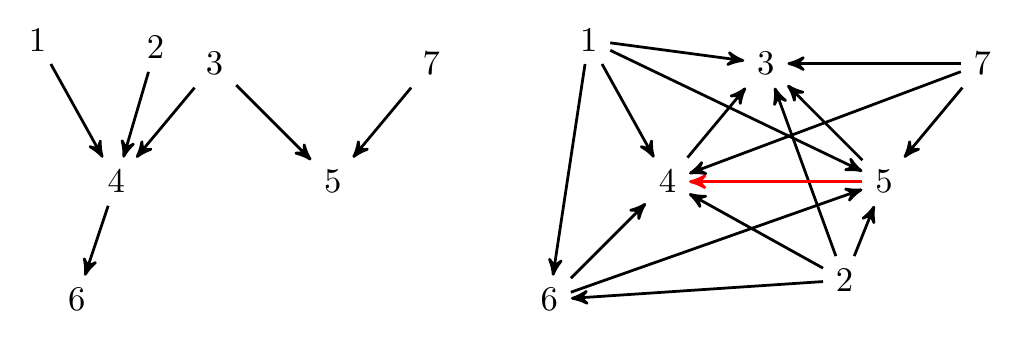}
	\caption{A generating DAG (left) and its GIES local maxima (right) for which GIES is not consistent.}
	\label{fig: counterexample}
	\end{figure}
The question whether GIES is consistent, was left open in~\cite{HB12}. We now prove that GIES is in general not consistent; i.e., if $n_j$ i.i.d.~samples are drawn from the interventional distribution $\PP^j$, then even as $n_1+\cdots+n_K\to\infty$ and under the faithfulness assumption, GIES may not recover the optimal $\I$-MEC with probability $1$.  
Consider the data-generating DAG depicted on the left in Figure~\ref{fig: counterexample}.  
Suppose we take interventions $\I$ consisting of $I_1 = \emptyset,I_2 = \{4\}, I_3 = \{5\}$, and that GIES arrives at the DAG $\G$ depicted on the right in Figure~\ref{fig: counterexample}.  
If the data collected grows as $n_1 = Cn_2 = Cn_3$ for some constant $C>1$, then we can show that the BIC score of $\G$ is a local maximum with probability $\frac{1}{2}$ as $n_1$ tends to infinity.  
The proof of this fact relies on the observation that GIES must initialize the turning phase at $\G$, and that $\G$ contains precisely one covered arrow $5\rightarrow 4$, which is colored red in Figure~\ref{fig: counterexample}.  
The full proof is given in the Supplementary Material.  

{\bf Greedy SP.}
In this paper we adapt the hybrid algorithm Greedy SP to provide consistent algorithms that use both interventional and observational data.  Greedy SP is a permutation-based algorithm that associates a  DAG to every permutation of the random variables and greedily updates the DAG by transposing elements of the permutation.  
More precisely, given a set of observed CI relations $\CC$ and a permutation $\pi = \pi_1\cdots\pi_p$, the Greedy SP algorithm assigns a DAG $\G_\pi := ([p],A_\pi)$ to $\pi$ via the rule
$$
\pi_i\rightarrow\pi_j\in A_\pi \quad \Longleftrightarrow \quad i<j \mbox{ and } \pi_i\not\independent \pi_j \, \mid \, \{\pi_1,\ldots,\pi_{\max(i,j)}\}\backslash\{\pi_i,\pi_j\},
$$
for all $1\leq i<j\leq p$.
The DAG $\G_\pi$ is a \emph{minimal I-MAP} (independence map) with respect to $\CC$, since any DAG $\mathcal{G}_\pi$ is Markov with respect to $\CC$ and any proper subDAG of $\mathcal{G}_\pi$ encodes a CI relation that is not in $\CC$~\cite{Pearl_1988}. 
Using a depth-first search approach, the algorithm reverses covered edges in $\G_\pi$, takes a linear extension $\tau$ of the resulting DAG and re-evaluates against $\CC$ to see if $\G_\tau$ has fewer arrows than $\G_\pi$.  
If so, the algorithm reinitializes at $\tau$, and repeats this process until no sparser DAG can be recovered.  
In the observational setting, Greedy SP is known to be consistent whenever the data-generating distribution is faithful to the sparsest DAG \cite{SWUM17}.  

\section{Two Permutation-Based Algorithms with Interventions} 
\label{sec: two permutation-based algorithms}


We now introduce our two interventional adaptations of Greedy SP and prove that they are consistent under the faithfulness assumption. 
In the first algorithm, presented in Algorithm~\ref{alg: yuhaos bic}, we use the same moves as Greedy SP, but we optimize with respect to a new score function that utilizes interventional data, namely the sum of the interventional BIC scores.  To be more precise,
for a collection of interventions $\I = \{I_1,\ldots,I_K\}$, the new score function is
\begin{equation*}
\label{eq:score}
\Score(\G) := \sum_{k = 1}^K \left( \underset{(A, \Omega) \in \G^k}{\M}\; \ell_{k} \left( \hat{X}^k; A, \Omega \right) \right) - \sum_{k=1}^K \lambda_{n_k} \vert \G^k \vert, 
\end{equation*}
where $\ell_k$ denotes the log-likelihood of the interventional distribution $\PP^k$, $(A,\Omega)$ are any parameters consistent with $\G^k$, $|\G|$ denotes the number of arrows in $\G$, and $\lambda_{n_k}  = \frac{\log n_k}{n_k}$.
\begin{algorithm}[t!]
\caption{}
\label{alg: yuhaos bic}
\KwInput{Observations $\hat{X}$, an initial permutation $\pi_0$, a threshold $\delta_n > \sum_{k=1}^K\lambda_{n_k}$, and a set of interventional targets $\I = \{I_1,\ldots,I_K\}$.}
\KwOutput{A permutation $\pi$ and its minimal I-MAP $\G_{\pi}$.}
\BlankLine
    Set $\G_\pi := \underset{\G\, \text{consistent with}\, \pi}{\argmax} \Score(\G)$\;
    Using a depth-first search approach with root $\pi$, search for a permutation $\pi_s$ with 
    $
    \Score(\G_{\pi_s}) > \Score(\G_\pi)
    $
    that is connected to $\pi$ through a sequence of permutations 
    $$
    \pi_0 = \pi,\pi_1, \cdots, \pi_{s-1},\pi_s,
    $$ 
    where each permutation $\pi_k$ is produced from $\pi_{k-1}$ by a transposition that corresponds to a covered edge in $\G_{\pi_{k-1}}$ such that 
    $
    \Score(\G_{\pi_k})  > \Score(\G_{\pi_{k-1}}) - \delta_n.
    $
    If no such $\G_{\pi_s}$ exists, return $\pi$ and $\G_\pi$; else set $\pi :=\pi_s$ and repeat.
\end{algorithm}

When Algorithm~\ref{alg: yuhaos bic} has access to observational and interventional data, then uniform consistency follows using similar techniques to those used to prove uniform consistency of Greedy SP in~\cite{SWUM17}. A full proof of the following consistency result for Algorithm~\ref{alg: yuhaos bic} is given in the Supplementary Material.
\begin{theorem}
\label{thm: uniform consistency of yuhaos bic}
Suppose $\PP$ is Markov with respect to an unknown I-MAP $\G_{\pi^\ast}$.  
Suppose also that observational and interventional data are drawn from $\PP$ for a collection of interventional targets $\I = \{I_1 := \emptyset, I_2,\ldots, I_K\}$.   
If $\PP^k$ is faithful to $(\G_{\pi^\ast})^k$ for all $k\in[K]$, then Algorithm~\ref{alg: yuhaos bic} returns the $\I$-MEC of the data-generating DAG $\G_{\pi^\ast}$ almost surely as $n_k \to \infty$ for all $k \in [K]$.  
\end{theorem}

A problematic feature of Algorithm~\ref{alg: yuhaos bic} from a computational perspective is the  the slack parameter~$\delta_n$.  
In fact, if this parameter were not included, then Algorithm~\ref{alg: yuhaos bic} would not be consistent.  
This can be seen via an application of Algorithm~\ref{alg: yuhaos bic} to the example depicted in Figure~\ref{fig: counterexample}.  
Using the same set-up as the inconsistency example for GIES, suppose that the left-most DAG $\G$ in Figure~\ref{fig: counterexample} is the data generating DAG, and that we draw $n_k$ i.i.d.~samples from the interventional distribution $\PP^k$ for the collection of targets $\I = \{\I_1 = \emptyset,\I_2 = \{4\}, \I_3 = \{5\}\}$.  
Suppose also that $n_1 = Cn_2  = Cn_3$ for some constant $C>1$, and now additionally assume that we initialize Algorithm~\ref{alg: yuhaos bic} at the permutation 
$
\pi = 1276543.
$
Then the minimal I-MAP $\G_\pi$ is precisely the DAG presented on the right in Figure~\ref{fig: counterexample}.  
This DAG contains one covered arrow, namely $5 \rightarrow 4$. Reversing it produces the minimal I-MAP $\G_\tau$ for 
$
\tau = 1276453.
$
Computing the score difference $\Score(\G_\tau) - \Score(\G_\pi)$ using \cite[Lemma~5.1]{NHM15} shows that as $n_1$ tends to infinity, $\Score(\G_\tau) < \Score(\G_\pi)$ with probability $\frac{1}{2}$. Hence, Algorithm~\ref{alg: yuhaos bic} would not be consistent without the slack parameter $\delta_n$. This calculation can be found in the Supplementary Material. 

Our second interventional adaptation of the Greedy SP algorithm, presented in Algorithm~\ref{alg: interventional greedy sp}, leaves the score function the same (i.e., the number of edges of the minimal I-MAP), but restricts the possible covered arrow reversals that can be queried at each step.  
In order to describe this restricted set of moves we provide the following definitions.
\begin{definition}
\label{def:covered}
Let $\I = \{I_1,\ldots,I_K\}$ be a collection of interventions, and for $i,j\in[p]$ define the collection of indices
$$
\I_{i \setminus j} := \lbrace k\in[K] : i \in I_k\ \textrm{ and }\ j \not\in I_k \rbrace. 
$$
For a minimal I-MAP $\G_\pi$ we say that a covered arrow $i\rightarrow j\in \G_\pi$ is \emph{$\I$-covered} if
\begin{align*}
\I_{i \setminus j} = \emptyset 
\quad
\textrm{ or }\quad i\rightarrow j \not\in (\G^k)_\pi \quad\textrm{ for all } \,k \in \I_{i \setminus j}.
\end{align*}
\end{definition}
\begin{definition}
\label{def: I-contradicting arrows}
We say that an arrow $i\rightarrow j \in \G_\pi$\, is\, \emph{$\I$-contradicting} if the following three conditions hold: 
(a)\,  $\I_{i \setminus j} \cup \I_{j \setminus i} \neq \emptyset$, \, (b)\, $\I_{i \setminus j} = \emptyset\, \textrm{ or }\, i \independent j \, \textrm{ in distribution } \,\PP^k\, \textrm{ for all } \,k \in \I_{i \setminus j}$, \,(c)\, $\I_{j \setminus i} = \emptyset\, \textrm{ or there exists }  \,k \in \I_{j \setminus i} \,\textrm{ such that }\, i \not\independent j \, \textrm{ in distribution } \,\PP^k$. 
\end{definition}
\begin{algorithm}[t!]
\caption{Interventional Greedy SP (IGSP)}
\label{alg: interventional greedy sp}
\KwInput{A collection of interventional targets $\I = \{\I_1,\ldots, \I_K\}$ and a starting permutation $\pi_0$.}
\KwOutput{A permutation $\pi$ and its minimal I-MAP $\G_\pi$.}
\BlankLine
    Set $\G := \G_{\pi_0}$\;
    Using a depth-first-search approach with root $\pi$, search for a minimal I-MAP $\G_\tau$ with $\vert \G \vert > \vert \G_\tau \vert$ that is connected to $\G$ by a list of $\I$-covered edge reversals. Along the search, prioritize the $\I$-covered edges that are also $\I$-contradicting edges. If such $\G_\tau$ exists, set $\G := \G_\tau$, update the number of $\I$-contradicting edges,  and repeat this step. If not, output $G_\tau$ with $\vert \G \vert = \vert \G_\tau \vert$ that is connected to $\G$ by a list of $\I$-covered edges and minimizes the number of $\I$-contradicting edges.
\end{algorithm}

In the observational setting, GES and Greedy SP utilize covered arrow reversals to transition between members of a single MEC as well as between MECs~\cite{C95,C02,SWUM17}.  
Since an $\I$-MEC is characterized by the skeleta and immoralities of each of its interventional DAGs, $\I$-covered arrows represent the natural candidate for analogous transitionary moves between $\I$-MECs in the interventional setting.  
It is possible that reversing an $\I$-covered edge $i\rightarrow j$ in a minimal I-MAP $\G_\pi$ results in a new minimal I-MAP $\G_\tau$ that is in the same $\I$-MEC as $\G_\pi$.  
Namely, this happens when $i\rightarrow j$ is a non-$\I$-essential edge in $\G_\pi$.  
Similar to Greedy SP, Algorithm~\ref{alg: interventional greedy sp} implements a depth-first-search approach that allows for such $\I$-covered arrow reversals, but it prioritizes those $\I$-covered arrow reversals that produce a minimal I-MAP $\G_\tau$ that is not $\I$-Markov equivalent to $\G_\pi$. 
These arrows are $\I$-contradicting arrows.  
The result of this refined search via $\I$-covered arrow reversal is an algorithm that is consistent under the faithfulness assumption. 
\begin{theorem}
\label{thm: consistency of interventional greedy sp}
Algorithm~\ref{alg: interventional greedy sp} is consistent under the faithfulness assumption.  
\end{theorem}

The proof of Theorem~\ref{thm: consistency of interventional greedy sp} is given in the Supplementary Material. When only observational data is available, Algorithm~\ref{alg: interventional greedy sp} boils down to greedy SP. We remark that the number of queries conducted in a given step of Algorithm~\ref{alg: interventional greedy sp} is, in general, strictly less than in the purely observational setting.  
That is to say, $\I$-covered arrows generally constitute a strict subset of the covered arrows in a DAG.  
At first glance, keeping track of the $\I$-covered edges may appear computationally inefficient.  
However, at each step we only need to update this list locally; so the computational complexity of the algorithm is not drastically impacted by this procedure. Hence, access to interventional data is beneficial in two ways: it allows to reduce the search directions at every step and it often allows to estimate the true DAG more accurately, since an $\I$-MEC is in  general smaller than an MEC.
Note that in this paper all the theoretical analysis are based on the low-dimensional setting, where $p \ll n$. The high-dimensional consistency of greedy SP is shown in~\cite{SWUM17}, and it is not difficult to see that the same high-dimensional consistency guarantees also apply to IGSP.


\begin{ideas}
\end{ideas}

\section{Evaluation} 
\label{sec: evaluation}

In this section, we compare Algorithm~\ref{alg: interventional greedy sp}, which we call \emph{Interventional Greedy SP} (\emph{IGSP}) with GIES on both simulated and real data. Algorithm~\ref{alg: yuhaos bic} is of interest from a theoretical perspective, but it is computationally inefficient since it requires performing two variable selection procedures per update.  Therefore, it will not be analyzed in this section.
The code utilized for the following experiments can be found at \url{https://github.com/yuhaow/sp-intervention}.

\subsection{Simulations}
\label{subsec: simulations}

Our simulations are conducted for linear structural equation models with Gaussian noise:
$$
(X_1,\ldots,X_p)^T = ((X_1,\ldots,X_p)A)^T +\epsilon,
$$
where $\epsilon\sim\mathcal{N}(0,{\bf 1}_p)$ and $A = (a_{ij})_{i,j = 1}^p$ is an upper-triangular matrix of edge weights with $a_{ij}\neq 0$ if and only if $i\rightarrow j$ is an arrow in the underlying DAG $\G^\ast$. 
For each simulation study we generated $100$ realizations of an (Erd\"os-Renyi) random $p$-node Gaussian DAG model for  $p\in\{10, 20\}$ with an expected edge density of 1.5.  
The collections of interventional targets $\I = \{I_0 := \emptyset, I_1,\ldots,I_K\}$ always consist of the empty set $I_0$ together with $K = 1$ or $2$. 
For $p = 10$, the size of each intervention set was $5$ for $K=1$ and $4$ for $K=2$. For $p = 20$, the size was increased up to $10$ and $8$ to keep the proportion of intervened nodes constant.
In each study, we compared GIES with Algorithm~\ref{alg: interventional greedy sp} for $n$ samples for each intervention with $n = 10^3,10^4,10^5$. 
\begin{figure}[t!]
\centering
\subfigure[$p=10$, $K = 1$]{\includegraphics[width=0.245\textwidth]{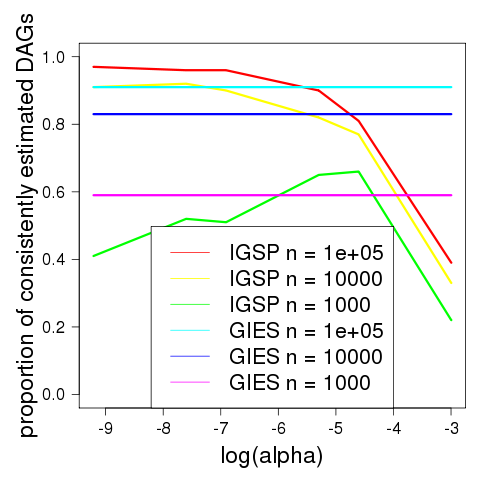}\label{fig:10_3}}
\subfigure[$p=10$, $K = 2$]{\includegraphics[width=0.245\textwidth]{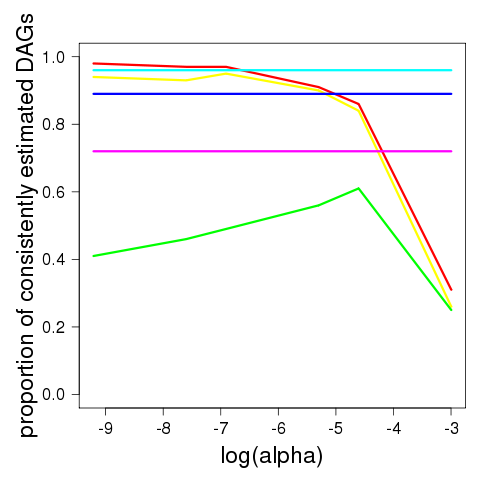}\label{fig:10_5}}
\subfigure[$p=20$, $K = 1$]{\includegraphics[width=0.245\textwidth]{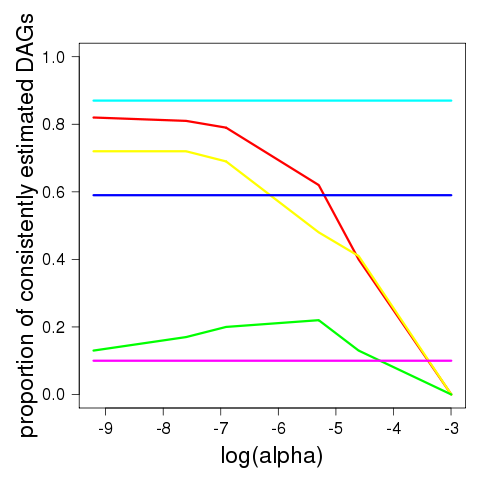}\label{fig: 20_3}}
\subfigure[$p=20$, $K = 2$]{\includegraphics[width=0.245\textwidth]{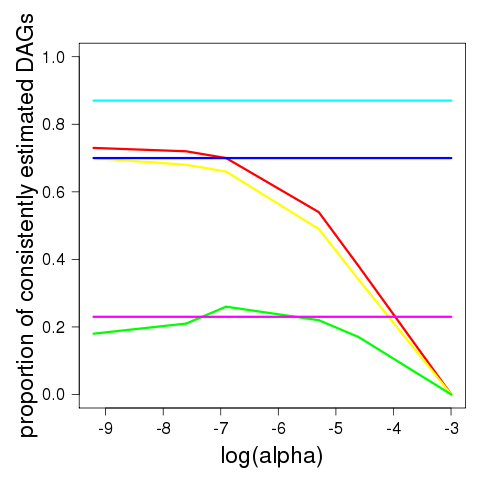}\label{fig: 20_5}}
\caption{
The proportion of consistently estimated DAGs for 100 Gaussian DAG models on $p$ nodes with $K$ single-node interventions.  }
	\label{fig: simulations for 10 and 20 nodes}
\end{figure}
Figure~\ref{fig: simulations for 10 and 20 nodes} shows the proportion of consistently estimated DAGs as distributed by choice of cut-off parameter for partial correlation tests. Interestingly, although GIES is not consistent on random DAGs, in some cases it performs better than IGSP, in particular for smaller sample sizes. 
However, as implied by the consistency guarantees given in Theorem~\ref{thm: consistency of interventional greedy sp}, IGSP performs better as the sample size increases.

We also conducted a focused simulation study on models for which the data-generating DAG $\G$ is that depicted on the left in Figure~\ref{fig: counterexample}, for which GIES is not consistent.  
In this simulation study, we took 100 realizations of Gaussian models for the data-generating DAG $\G$ for which the nonzero edge-weights $a_{ij}$ were randomly drawn from $[-1,-c,)\cup(c,1]$ for $c = 0.1,0.25,0.5$. The interventional targets were $\I = \lbrace I_0 = \emptyset$, $I_1 \rbrace$, where $I_1$ was uniformly at random chosen to be $\lbrace 4 \rbrace$, $\lbrace 5 \rbrace$, $\lbrace 4,5 \rbrace$. 
Figure~\ref{fig: counterexample simulations} shows, for each choice of $c$, the proportion of times $\G$ was consistently estimated as distributed by the choice of cut-off parameter for the partial correlation tests.  
We see from these plots that as expected from our theoretical results GIES recovers $\G$ at a lower rate than Algorithm~\ref{alg: interventional greedy sp}. 
\begin{figure}[t!]
\centering
\subfigure[$c = 0.1$]{\includegraphics[width=0.32\textwidth]{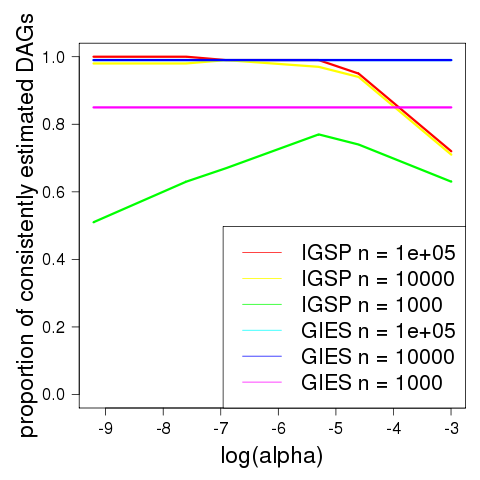}\label{fig:10_3}}
\subfigure[$c = 0.25$]{\includegraphics[width=0.32\textwidth]{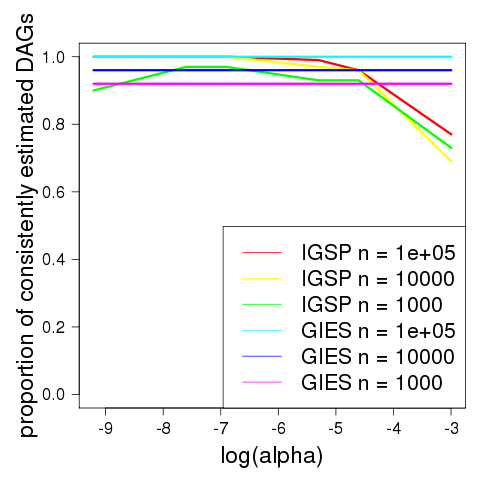}\label{fig:10_5}}
\subfigure[$c = 0.5$]{\includegraphics[width=0.32\textwidth]{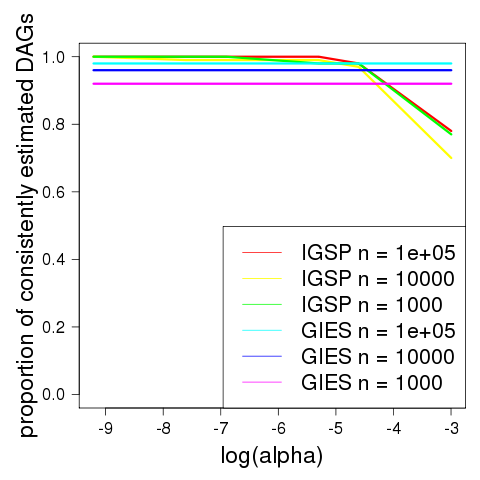}\label{fig: 20_3}}
\caption{
Proportion of times the DAG $\G$ from Figure~\ref{fig: counterexample} (left) is consistently estimated under GIES and Algorithm~\ref{alg: interventional greedy sp} for Gaussian DAG models with edge-weights drawn from $[-1,-c)\cup(c,1]$.
}
	\label{fig: counterexample simulations}
\end{figure}

\subsection{Application to Real Data}
\label{subsec: application on real data}
In the following, we report results for studies conducted on two real datasets coming from genomics.  
The first dataset is the protein signaling dataset of \emph{Sachs et al.} \cite{Sachs_2005}, and the second is the single-cell gene expression data generated using perturb-seq in \cite{perturb-seq}.

{\bf Analysis of protein signaling data.}  
The dataset of Sachs et al.~\cite{Sachs_2005} consists of 7466 measurements of the abundance of phosphoproteins and phospholipids recorded under different experimental conditions in primary human immune system cells.  
The different experimental conditions are generated using various reagents that inhibit or activate signaling nodes, and thereby correspond to interventions at different nodes in the protein signaling network.  
The dataset is purely interventional and most interventions take place at more than one target.  
Since some of the experimental perturbations effect receptor enzymes instead of the measured signaling molecules, we consider only the 5846 measurements in which the perturbations of receptor enzymes are identical.  
In this way, we can define the observational distribution to be that of molecule abundances in the model where only the receptor enzymes are perturbed.  
This results in $1755$ observational measurements and $4091$ interventional measurements.  Table~E.2 in the Supplementary Material summarizes the number of samples as well as the targets for each intervention.
For this dataset we compared the GIES results reported in \cite{HB15} with Algorithm~\ref{alg: interventional greedy sp} using both, a linear Gaussian and a kernel-based independence criterium \cite{F07,T09}.  
A crucial advantage of Algorithm~\ref{alg: interventional greedy sp} over GIES is that it is nonparametric and does not require Gaussianity.  
In particular, it supports kernel-based CI tests that are in general able to deal better with non-linear relationships and non-Gaussian noise, a feature that is typical of datasets such as this one.  
\begin{figure}[t!]
\centering
\subfigure[Directed edge recovery]{\includegraphics[width=0.45\textwidth]{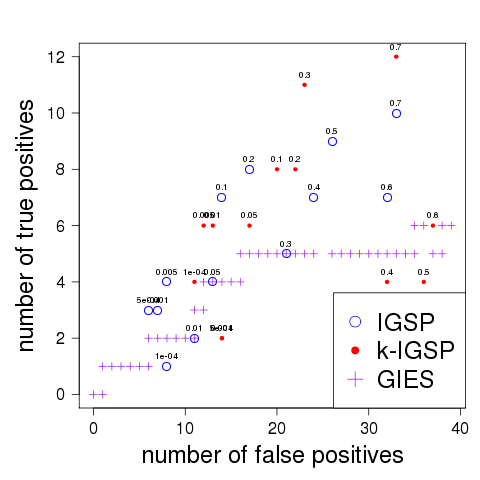}\label{fig: TPFP}} \qquad
\subfigure[Skeleton recovery]{\includegraphics[width=0.45\textwidth]{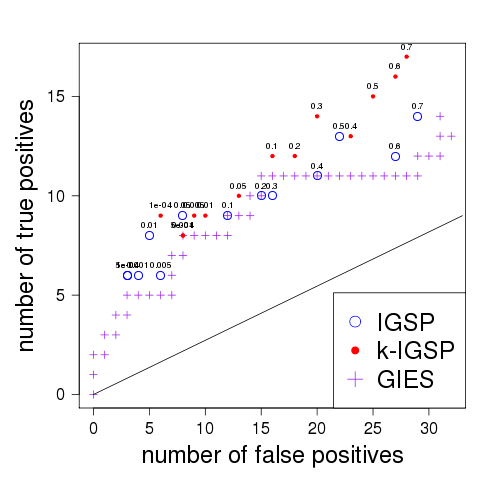}\label{fig: TPFP_Skeleton}}
\caption{
ROC plot of the models estimated from the data~\cite{Sachs_2005} using GIES as reported in~\cite{HB12} and the linear Gaussian and kernel-based versions of IGSP with different cut-off values for the CI tests. The solid line indicates the accuracy achieved by random guessing.
}
	\label{fig: sachs et al}
\end{figure}

For the GIES algorithm we present the results of \cite{HB12} in which the authors varied the number of edge additions, deletions, and reversals as tuning parameters.  
For the linear Gaussian and kernel-based implementations of IGSP our tuning parameter is the cut-off value for the CI tests, just as in the simulated data studies in Section~\ref{subsec: simulations}.  
Figure~\ref{fig: sachs et al} reports our results for thirteen different cut-off values in $[10^{-4}, 0.7]$, which label the corresponding points in the plots.  
The linear Gaussian and kernel-based implementations of IGSP are comparable and generally both outperform GIES.  
The Supplementary Material contains a comparison of the results obtained by IGSP on this dataset to other recent methods that allow also for latent confounders, such as ACI, COmbINE and ICP.

{\bf Analysis of perturb-seq gene expression data.}
We analyzed the performance of GIES and IGSP on perturb-seq data published by Dixit et al. \cite{perturb-seq}. The dataset contains observational data as well as interventional data from $\sim$30,000 bone marrow-derived dendritic cells (BMDCs). Each data point contains gene expression measurements of 32,777 genes, and each interventional data point comes from a cell where a single gene has been targeted for deletion using the CRISPR/Cas9 system.

After processing the data for quality, the data consists of 992 observational samples and 13,435 interventional samples from eight gene deletions. The number of samples collected under each of the eight interventions is shown in the Supplementary Material. These interventions were chosen based on empirical evidence that the gene deletion was effective\footnote{An intervention was considered effective if the distribution of the gene expression levels of the deleted gene is significantly different from the distribution of its expression levels without intervention, based on a Wilcoxon Rank-Sum test with $\alpha$ = 0.05. Ineffective interventions on a gene are typically due to poor targeting ability of the guide-RNA designed for that gene.}. We used GIES and IGSP to learn causal DAGs over 24 of the measured genes, including the ones targeted by the interventions, using both observational and interventional data. We followed \cite{perturb-seq} in focusing on these 24 genes, as they are general transcription factors known to regulate each other as well as numerous other genes~\cite{garber12}.

We evaluated the learned causal DAGs based on their accuracy in predicting the true effects of each of the interventions (shown in Figure \ref{fig: heatmap}) when leaving out the data for that intervention. Specifically, if the predicted DAG indicates an arrow from gene A to gene B, we count this as a true positive if knocking out gene A caused a significant change\footnote{Based on a Wilcoxon Rank-Sum test with $\alpha$ = 0.05, which is approximately equivalent to a q-value of magnitude $\geq$ 3 in Figure \ref{fig: heatmap}} in the distribution of gene B, and a false positive otherwise. 
For each inference algorithm and for every choice of the tuning parameters, we learned eight causal DAGs, each one trained with one of the interventional datasets being left out. We then evaluated each algorithm based on how well the causal DAGs are able to predict the corresponding held-out interventional data. As seen in Figure~\ref{fig: TPFP_perturb_seq}, IGSP predicted the held-out interventional data better than GIES (as implemented in the R-package pcalg) and random guessing, for a number of choices of the cut-off parameter. The true and reconstructed networks for both genomics datasets are shown in the Supplementary Material.

\section{Discussion}
\label{sec: discussion}
We have presented two hybrid algorithms for causal inference using both observational and interventional data and we proved that both algorithms are consistent under the faithfulness assumption.  
These algorithms are both interventional adaptations of the Greedy SP algorithm and are the first algorithms of this type that have consistency guarantees.  
While Algorithm~\ref{alg: yuhaos bic} suffers a high level of inefficiency, IGSP is implementable and competitive with the state-of-the-art, i.e., GIES.  
Moreover, IGSP has the distinct advantage that it is nonparametric and therefore does not require a linear Gaussian assumption on the data-generating distribution.  
We conducted real data studies for protein signaling and single-cell gene expression datasets, which are typically non-linear with non-Gaussian noise.  
In general, IGSP outperformed GIES.  
This purports IGSP as a viable method for analyzing the new high-resolution datasets now being produced by procedures such as perturb-seq. 
An important challenge for future work is to make these algorithms scale to 20,000 nodes, i.e., the typical number of genes in such studies. 
In addition, in future work it would be interesting to extend IGSP to allow for latent confounders. An advantage of not allowing for latent confounders is that a DAG is usually more identifiable. For example, if we consider a DAG with two observable nodes, a DAG without confounders is fully identifiable by intervening on only one of the two nodes, but the same is not true for a DAG with confounders.

\begin{figure}[t!]
\centering
\subfigure[True effects of gene deletions]{\includegraphics[width=0.45\textwidth]{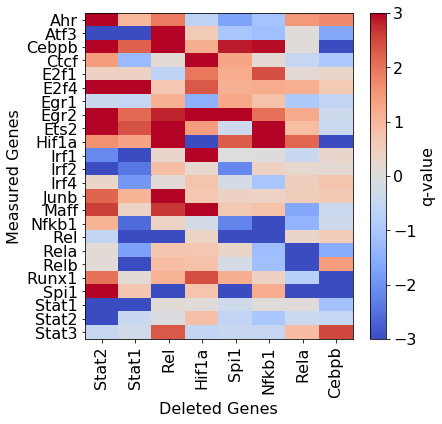}\label{fig: heatmap}}
\subfigure[Causal effect prediction accuracy rate]{\includegraphics[width=0.45\textwidth]{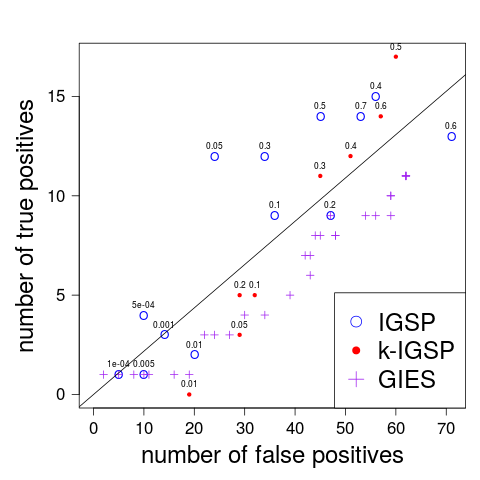}\label{fig: TPFP_perturb_seq}}
\caption{
(a) Heatmap of the true effects of each gene deletion on each measured gene. The q-value has the same magnitude as the log p-value of the Wilcoxon rank-sum test between the distributions of observational data and the interventional data. Positive and negative q-values indicate increased and decreased abundance as a result of deletion respectively. (b) ROC plot of prediction accuracy by the causal DAGs learned by IGSP and GIES. The solid line indicates the accuracy achieved by random guessing. 
}
  \label{fig:perturb_seq}
\end{figure}

\section*{Acknowledgements}
Yuhao Wang was supported by DARPA (W911NF-16-1-0551) and ONR (N00014-17-1-2147). Liam Solus was supported by an NSF Mathematical Sciences Postdoctoral Research Fellowship (DMS - 1606407). Karren Yang was supported by the MIT Department of Biological Engineering. Caroline Uhler was partially supported by DARPA (W911NF-16-1-0551), NSF (1651995) and ONR (N00014-17-1-2147). We thank Dr.~Sofia Triantafillou from the University of Crete for helping us run COmbINE.

\newpage
\appendix

\section*{\LARGE{Appendix}}
\label{sec: appendix}

\counterwithin{figure}{section}
\counterwithin{table}{section}

\section{Counterexample to Consistency of GIES.}

In the following, we verify that the example of GIES described in Section~3 is in fact a counterexample to consistency of GIES with the BIC score function.  
Recall that the DAG on the left in Figure~1, which we denote $\G_0$, is taken to be the data-generating DAG, and our collection of interventions is $\I = \{I_1 = \emptyset, I_2 = \{4\}, I_3 = \{5\}\}$.  
Suppose that the number of samples, $n_1,n_2,n_3$, drawn from the interventional distributions $\PP^1,\PP^2,\PP^3$, satisfy $n_1 = Cn_2 = Cn_3$ for some constant $C>1$, and that GIES arrives at the DAG $\G$ depicted on the right in Figure~1.  
Here, we also assume that the observational distribution $\PP^1$ is faithful to $\G_0$.  
We claim that this DAG is a local maximum of the GIES algorithm.  

To see this, first notice that since $5\rightarrow 4$ is the only covered edge in $\G$, then its $\I$-MEC has size one.  
Also notice that the DAG $\G$ is the minimal I-MAP of $\G_\pi$ for the permutation 
$
\pi = 1276543.  
$
Therefore, by consistency of GES under faithfulness~\cite{C02}, deleting any edge of $\G$ would result in a DAG with a strictly higher BIC.  
Thus, it only remains to verify that $\G$ is a local maximum with respect to the turning phase.  
We begin by checking that turning the only covered arrow in $\G$ does not increase the BIC score function with probablilty $1$.  
In the following, for a node $j$, we let
$
\I_{-j} := \I\backslash \lbrace k \mid j \in I_k\rbrace.
$
We may then express the score of $\G$ as 
\begin{align*}
\Score(\G, \hat{X}) := \sum_{j=1}^p s(j, {\pa}_j(\G), \hat{X}^{\I_{-j}}) + C - \lambda_n \vert \G \vert, 
\end{align*}
where $s(j, {\pa}_j(\G), \hat{X}^{\I_{-j}})$ is the log of the regression residual when regressing $j$ on ${\pa}_j(\G)$ using the data from the truncated intervention set $\I_{-j}$.  
Formally, 
\begin{align*}
s(j, {\pa}_j(\G), \hat{X}^{\I_{-j}}) = - \frac{1}{2} \frac{n_{-j}}{n} \log \left( \underset{a \in \R^{\vert {\pa}_j(\G) \vert}}{\min} \sum_{k \in \I_{-j}}\Vert \hat{X}_j^k - \hat{X}_{{\pa}_j(\G)}^k \cdot a \Vert_2^2 / n_{-j} \right).
\end{align*}
Now let $\G^\prime$ denote the DAG produced by reversing the arrow $d\rightarrow c$ in $\G$, and let $\hat{\rho}_{i, j \mid S}^\I$ denote the partial correlation testing coefficient of $i$ and $j$ given some $S\subset[p]$ using interventional data $\hat{X}^k, \forall k \in \I$.
If we take $S = \parents_\G(i)\cap\parents_\G(j)$, then by~\cite[Lemma 5.1]{NHM15} we have that
\begin{align*}
\begin{split}
\Score(\G', \hat{X})  - \Score(\G, \hat{X}) &  = s(c, S, \hat{X}^{\I_{-c}}) - s(c, S \union \lbrace d \rbrace, \hat{X}^{\I_{-c}}) \\
& \hspace{.5in} + s(d, S \union \lbrace c \rbrace, \hat{X}^{\I_{-d}}) - s(d, S, \hat{X}^{\I_{-d}}), \\
& = - \frac{1}{2} \frac{n_{-d}}{n} \log(1 - (\hat{\rho}_{c, d \mid S}^{\I_{-d}})^2) + \frac{1}{2} \frac{n_{-c}}{n} \log(1 - (\hat{\rho}_{c, d \mid S}^{\I_{-c}})^2).
\end{split}
\end{align*}
Since $n_1 = Cn_2 = Cn_3$ it follows that the distributions of $\hat{\rho}_{c, d \mid S}^{\I_{-d}}$ and $\hat{\rho}_{c, d \mid S}^{\I_{-c}}$ are always identical.  
Therefore, $\Score(\G', \hat{X}) < \Score(\G, \hat{X})$ with probability $\frac{1}{2}$.

It now only remains to verify that turning any non-covered edge in $\G$ increases the value of the BIC score function.  
Suppose that $\G^\prime$ is a DAG produced from turning some edge in $\G$ other than $5\rightarrow 4$.  
Since such an edge is not covered, $\G^\prime$ will not be an independence map of the un-intervened distribution $\PP^1$.  
Therefore, there exists some sufficiently large $C>1$ such that the score of $\G$ is larger than $\G^\prime$.  
This is because the score function for $C$ large enough is dominated by the part that depends on the observational data.
Thus, we conclude that $\G$ is a local maximum of GIES. 
\hfill $\square$

\bigskip

\section{Counterexample to Consistency of Algorithm~1 without the Slack Factor}

We now verify that the example described in Section~4 shows that Algorithm~1 without the use of the slack factor $\delta_n$ is not consistent.  
The proof of this statement is similar to that of the counterexample to consistency of GIES, and so we adopt the exact same set-up and notation used in the previous proof.  
Unlike GIES, Algorithm~1 only uses moves corresponding to reversals of covered edges in the observational DAG $\G_0$, depicted on the right in Figure~1.  
Thus, the only possible move Algorithm~1 can make is to reverse the covered arrow $5\rightarrow 4$.  
If we denote the resulting graph by $\G^\prime$, then similar to the previous proof, the difference in the scores $\Score(\G^\prime) - \Score(\G)$ can be computed as follows:
\begin{align*}
\begin{split}
\Score(\G', \hat{X})  - \Score(\G, \hat{X}) 	&= 	- \frac{1}{2} \sum_{k \in \I_{j \setminus i}} \log \left( 1 - (\hat{\rho}_{i, j \vert S}^{k})^2 \right) + \frac{1}{2} \sum_{k \in \I_{i \setminus j}} \log \left( 1 - (\hat{\rho}_{i, j \vert S}^{k})^2 \right)
\end{split}
\end{align*}
Since $n_1 = Cn_2 = Cn_3$ and there is no arrow between $4$ and $5$ in either of $\G^2$ or $\G^3$, then the distributions of $\hat{\rho}_{4, 5 \mid S}^{\I_{-5}}$ and $\hat{\rho}_{4, 5 \mid S}^{\I_{-4}}$ are identical.  
Therefore, $\Score(\G', \hat{X}) < \Score(\G, \hat{X})$ with probability $\frac{1}{2}$.  
\hfill$\square$

\bigskip

\section{Proof of Theorem~4.1}

Recall that a DAG $\HH$ is called an \emph{independence map} of a DAG $\G$, denoted $\G\leq\HH$, if every CI relation entailed by the $d$-separation statements of $\HH$ are also entailed by $\G$.  
The proof of Theorem~4.1 relies on the transformational relationship between a DAG and an independence map given in \cite[Theorem 4]{C02}.  
In short, the theorem states that for an independence map $\G\leq\HH$, there exists a sequence of covered arrow reversals and arrow additions such that after each arrow reversal or addition, the resulting DAG $\G^\prime$ satisfies $\G\leq\G^\prime\leq\HH$, and after all arrow reversals and additions $\G' = \HH$.  
The proof of this fact follows from the APPLY-EDGE OPERATION algorithm \cite{C02}, which describes the choices that can be made to produce such a transformation of independence maps.  
In \cite{SWUM17} the authors refer to the sequence of independence maps 
$$
\G\leq \G^{(1)}\leq\G^{(2)}\leq\cdots\leq\G^{(m-1)}\leq\HH
$$
that transforms $\G$ into $\HH$ as a \emph{Chickering sequence}.  

A key feature of the APPLY-EDGE OPERATION algorithm is that it recurses on the common sink nodes between $\G$ and $\HH$.  
Namely, if $\G$ and $\HH$ have any sink nodes with the same set of parents in both DAGs, the algorithm first deletes these nodes and compares the resulting subDAGs.  
This is repeated until there are no sink nodes in the two graphs with the exact same set of parents.
The remaining set of sink nodes that must be fixed are denoted $s_1\ldots, s_M$.  
Then the algorithm begins to reverse and add arrows to the relevant subDAG of $\G$ until a new common sink node appears, which it then deletes, and so on. 
Once the algorithm corrects one such sink node in the subDAG of $\G$ to match the same node in the subDAG of $\HH$, we say the sink node has been \emph{resolved}.   
In \cite{SWUM17} the authors observe that if we have an independence map of minimal I-MAPs $\G_\pi\leq\G_\tau$, then there exists a Chickering sequence that adds arrows and reverses arrows so that exactly one sink node is resolved at a time; i.e., there is no need to do arrow reversals and additions to any one sink node in order to resolve another.  
To prove Theorem~4.1, we utilize this fact and the following two lemmas. 
\begin{lemma} 
\label{lem: d-connecting equivalence}
Suppose $\G$ is an independence map of the data-generating DAG $\G_{\pi^\ast}$ for the permutation $\pi$. 
Let $i \to j$ denote a covered edge in $\G$, and let $S$ denote the set of nodes that precedes $i$ in permutation $\pi$; i.e.,
\begin{align*}
S := \lbrace \ell \mid \pi(\ell) < \pi(i) \rbrace,
\end{align*}
then in $\G_{\pi^\ast}$ the set of d-connecting paths from $i$ to $j$ given $S$ is the same as the set of d-connecting paths from $i$ to $j$ given $\parents_{\G}(i)$.
\end{lemma}
\begin{proof}
We prove this by contradiction.   
Suppose in $\G_{\pi^\ast}$ there exists a path $P_{i \to j}$ that d-connects $i$ and $j$ given $S$ but $P_{i \to j}$ is $d$-separated given $\parents_{\G}(i)$.  
Then there must exist at least one node $a \in S \setminus \parents_{\G}(i)$ that is a collider on $P_{i \to j}$ or a descendent of a collider on $P_{i \to j}$. 
If $a$ is a collider on $P_{i \to j}$, then $a$ $d$-connects $i$ given $S \setminus \lbrace a \rbrace$. 
If no such collider exists, then $a$ must be a descendent of a collider $s$ on $P_{i \to j}$.  
Moreover, $a$ $d$-connects $s$ given $S \setminus \lbrace a \rbrace$ and $s$ also $d$-connects $i$ given $S \setminus \lbrace a \rbrace$. 
Since $s \not\in S$, $a$ $d$-connects $i$ given $S \setminus \lbrace a \rbrace$. 
However, since $\G$ is an independence map, $a$ must be a parent of node $i$ in $\G$, which contradicts with the fact that $a \not\in \parents_{\G}(i)$.

Suppose in $\G_{\pi^\ast}$ there exists a path $P_{i \to j}$ that $d$-connects $i$ and $j$ given $\parents_{\G}(i)$ but is $d$-separated given $S$.  
Then there must exist some nodes in $S \setminus \parents_{\G}(i)$ that are non-colliders on $P_{i \to j}$. 
Let $a$ denote one of such nodes that is closest to $i$ on $P_{i \to j}$, then $a$ and $i$ must be d-connected given $S \setminus \lbrace a \rbrace$. 
Since $\G$ is an independence map, $a$ must be a parent of node $i$ in $\G$, which contradicts with the fact that $a \not\in \parents_{\G}(i)$.
\end{proof}

\begin{lemma}
\label{lem: d-connecting paths}
Given a permutation $\pi$ consider the sequence of minimal I-MAPs from $\G_\pi$ to the data-generating DAG $\G_{\pi^\ast}$ given by covered arrow reversals
\begin{align*}
\G_{\pi} = \G_{\pi^0} \geq \G_{\pi^1}\geq  \cdots \geq \G_{\pi^M} = \G_{\pi^\ast}.
\end{align*}
If the edge $i\rightarrow j$ is reversed in $\G_{\pi^{k-1}}$ to produce $\G_{\pi^k}$, then in $\G_{\pi^\ast}$ all d-connecting paths from $j$ to $i$ given $\parents_{\G_{\pi^{k-1}}}(i)$ must be pointing towards $i$ (i.e. the edge incident to $i$ on the path points to $i$).
\end{lemma}

\begin{proof}
By \cite[Theorem 15]{SWUM17}, we know there exists a Chickering sequence from $\G_{\pi^\ast}$ to $\G_{\pi}$
\begin{align*}
\G_{\pi^\ast} = \G^0\leq \G^1\leq \cdots\leq \G^N = \G_\pi
\end{align*}
that resolves one sink at a time and, respectfully, reverses one edge at a time. 
Let $s_1,\ldots, s_M$ denote the list of sink nodes resolved in the Chickering sequence, labeled so that $s_j$ is the $j^{th}$ sink node resolved in the sequence.  
More specifically, this means that the Chickering sequence can be divided in terms of a sublist of DAGs $\G^{i_1}, \cdots, \G^{i_M}$ such that
$\G^{i_j}$ is the DAG produced by resolving sink $s_j$.  
It follows that the DAGs in the subsequence 
\begin{align*}
\G^{i_{j-1}+1}\leq \cdots \leq \G^{i_{j}-1}
\end{align*}
correspond to the arrow additions and covered arrow reversals that are needed to resolve sink $s_{j}$.  
For $t = 1,\ldots, q_j$ let $z_t$ denote the node such that $s_j \rightarrow z_t$ must be reversed in order to produce $\G^{i_j}$ from $\G^{i_{j-1}}$.  
We label these nodes such that $s_j\rightarrow z_t$ is reversed before $s_j\rightarrow z_{t+1}$ in the given Chickering sequence.  
Let $\G^{i_{j,t}}$ denote the DAG generated after reversing edge $s_j \to z_{t}$
Then we can write our sequence $G^{i_{j-1}}\leq  \cdots\leq \G^{i_j}$ as:
\begin{align*}
G^{i_{j-1}}\leq \G^{i_{j-1}+1}\leq \cdots\leq \G^{i_{j,t}}\leq \cdots\leq \G^{i_{j,t+1}} \leq\cdots\leq \G^{i_{j,q_j} - 1}\leq \G^{i_{j,q_j}} = \G^{i_j}.
\end{align*}
To prove the lemma, we must then show that for all $j$ and $t$, all $d$-connecting paths in $\G_{\pi^\ast}$ from $s_j$ to $z_t$ given $\parents_{\G^{i_{j,t}}}(z_{t})$ are pointing towards $z_t$.  

To prove this, let $\pi^{j-1}$ denote a permutation consistent with $\G^{i_{j-1}}$ and let $S_{\pi^{j-1}}(z_{t})$ denote the set of nodes that precedes $z_{t}$ in the permutation $\pi^{j-1}$; i.e.,
\begin{align*}
S_{\pi^{j-1}}(z_{t}) := \lbrace \ell \mid \pi^{j-1}(\ell) < \pi^{j-1}(z_{t})) \rbrace.
\end{align*}
If $\pi^{j-1} = \ldots s_j \ldots z_1 \ldots z_t \ldots z_{q_j} \ldots$, then for $t = 1,\ldots,q_j$, we can always choose a linear extension $\pi^{j,t}$ of $\G^{i_{j,t}}$ in which $\pi^{j,t} = \ldots z_1 \ldots z_ts_j \ldots z_{q_j} \ldots$, and 
$S_{\pi^{j-1}}(z_{t})\backslash\{s_j\} = S_{\pi^{j, t}}(z_{t})$ by moving $s_j$ forward in the permutation $\pi^{j-1}$ until it directly follows $z_t$.  
It is always possible to pick such an extension as $\pi^{j,t}$ since we can choose the extension of $\pi^{j-1}$ so that the only nodes in between $z_{t-1}$ and $z_t$ are the descendants of $z_{t-1}$ that are also ancestors of $z_t$.  
The existence of such an ordering of $\pi$ with respect to the ordering of the nodes $z_1,\ldots,z_{q_j}$ is implied by the choice of the maximal child in each iteration of step 5 of the APPLY-EDGE OPERATION algorithm that produces the Chickering sequence~\cite{C02}.  
Using Lemma~\ref{lem: d-connecting equivalence}, 
we know that any $d$-connecting path from $z_t$ to $s_j$ given $S_{\pi^{j, t}}(z_{t})$ in $\G_{\pi^\ast}$ is actually the same as a $d$-connecting path from $z_t$ to $s_j$ in $\G_{\pi^\ast}$ given $\parents_{\G^{i_{j,t}}}(z_{t})$.  
Since $S_{{\pi^{j,t}}}(z_t) = S_{{\pi^{j-1}}}(z_t)\backslash\{s_j\}$ then it remains to show that any $d$-connecting path from $s_j$ to $z_t$ given $S_{{\pi^{j-1}}}(z_t)\backslash\{s_j\}$ in $\G_{\pi^\ast}$ goes to $z_t$.  
Since $s_j\rightarrow z_t$ in $\G_{\pi^{j-1}}$, we prove the following, slightly stronger, statement: for any edge $a\rightarrow b \in \tilde{\G}^{i_{j-1}}$, all d-connecting paths from $a$ to $b$ given $S_{\pi^{j-1}} (b) \setminus \lbrace a \rbrace$ in $\G_{\pi^\ast}$ go to $b$. 

We prove this stronger statement via induction. 
If $\tilde{\G}^{i_{j-1}} = \tilde{\G}_{\pi^\ast}$, the statement is definitely true since the only possible $d$-connection between $a$ and $b$ given $S_{\pi^\ast}(b) \setminus \lbrace a \rbrace$ is the arrow $a\rightarrow b \in \G_{\pi^\ast}$. 
Suppose it is also true when $j = j'-1$.  
Recall the only difference between $\pi^{j'-1}$ and $\pi^{j'}$ is the position of $s_{j'}$.  
If in $\tilde{\G}^{i_{j'}}$ there is a new arrow $a \to b$, then this arrow corresponds to some paths that $d$-connect $a$ and $b$ given $S_{\pi^{j'}}(b) \setminus \lbrace a \rbrace$.  
However, they are $d$-separated given $S_{\pi^{j'-1}}(b) \setminus \lbrace a \rbrace$. 
Since $S_{\pi^{j'}}(b) = S_{\pi^{j'-1}}(b) \setminus \lbrace s_{j'} \rbrace$, then $s_{j'}$ must be in the middle of these new paths and is not a collider. 
In this case, removing $s_{j'}$ from the conditioning set would turn these paths into $d$-connections. 

Without loss of generality, we consider one of these new paths from $a$ to $b$ denoted as $P_{a \to b}$. 
Since $s_{j'}$ is in the middle of $P_{a \to b}$, let $P_{s_{j'} \to b}$ denote the latter part of $P_{a \to b}$. 
Obviously, $P_{s_{j'} \to b}$ also $d$-connects $s_{j'}$ and $b$ given $S_{\pi^{j'-1}}(b) \setminus \lbrace s_{j'} \rbrace$. 
As $\G^{i_{j'-1}}$ is an independence map of $\G_{\pi^\ast}$, $s_{j'} \to b$ must be an edge in $\G^{i_{j'-1}}$, and therefore it also exists in $\tilde{\G}^{i_{j'-1}}$. 
Notice, if $s_{j'} \to b \in \tilde{\G}^{i_{j'-1}}$ then, in $\G_{\pi^\ast}$, all paths that $d$-connect $s_{j'}$ and $b$ given $S_{\pi^{j'-1}}(b) \setminus \lbrace s_{j'} \rbrace$ go to $b$.  
Therefore, $P_{s_j \to b}$ is a path that goes to $b$. 
In this case, $P_{a \to b}$ is also a path that goes to $b$. 
As there is no specification of $P_{a \to b}$, this holds for all new paths, and this completes the proof.
\end{proof}

\emph{Proof of Theorem~4.1.}
We can now prove Theorem~4.1.  
Let $\PP$ be a distribution that is faithful with respect to an unknown I-MAP $\G_{\pi^\ast}$.
Suppose that observational and interventional data are drawn from $\PP$ for a collection of interventional targets $\I = \{I_1 := \emptyset, I_2,\ldots, I_K\}$, and that $\PP^k$ is faithful to $\G_{\pi^\ast}^k$ for all $k\in[K]$.  
We must show that Algorithm~1 returns to $\I$-MEC of $\G_{\pi^\ast}$.  
Suppose that we are at the DAG $\G_\pi$ for some permutation $\pi$ of $[p]$.  
By \cite[Theorem 15]{SWUM17} there exists a sequence of minimal I-MAPS
$$
\G_{\pi^\ast} = \G_{\pi^M}\leq\G_{\pi^{M-1}}\leq \cdots\leq \G_{\pi^0} = \G_\pi,
$$
where $\G_{\pi^{k}}$ is produced from $\G_{\pi^{k-1}}$ by reversing a covered arrow $i\rightarrow j$ and then deleting some edges of $\G_{\pi^{k-1}}$.  
In particular, this sequence arises from a Chickering sequence that resolves one sink node at a time, as in Lemma~\ref{lem: d-connecting paths}.  
We would now like to see that for such a path 
$$
\Score(\G_{\pi^{k}}) \geq \Score(\G_{\pi^{k-1}})- \delta_n,
$$
for all $k = 1,2,\ldots, M$.    
Suppose first that $\G_{\pi^{k-1}}$ and $\G_{\pi^{k}}$ differ only by a covered arrow reversal (i.e., they have the same skeleton).  
Using the notation from the previous proofs, we let $\hat{\rho}_{i, j \mid S}^k$ denote the value of the partial correlation of $i, j \mid S$ for some set $S\subset[p]$ based on data $\hat{X}^k$ from the intervention $I_k$.  
If we take $S = \parents_i(\G_{\pi^{k-1}})$, then by \cite[Lemma 5.1]{NHM15} and Lemma~\ref{lem: d-connecting paths} it follows that
\begin{align*}
\begin{split}
\Score(\G_{\pi^{k}}) - \Score(\G_{\pi^{k-1}}) = &- \frac{1}{2} \sum_{k \in \I_{j \setminus i}} \left(\log \left( 1 - (\hat{\rho}_{i, j \vert S}^{k})^2 \right) + \lambda_{n_k}\right)  \\
	&+ \frac{1}{2} \sum_{k \in \I_{i \setminus j}} \left(\log \left( 1 - (\hat{\rho}_{i, j \vert S}^{k})^2 \right) + \lambda_{n_k}\right).\\
\end{split}
\end{align*}
Note that the value of $\underset{k \in \I_{i \setminus j}}{\sum} \left( \log \left( 1 - (\hat{\rho}_{i, j \vert S}^{k})^2 \right) + \lambda_{n_k} \right)$ will be zero when the set $\I_{i \setminus j}$ is empty.  
It then follows from Lemma~\ref{lem: d-connecting paths} that 
$
\Score(\G_{\pi^{k}}) \geq \Score(\G_{\pi^{k-1}})- \delta_n,
$
for all $k = 1,\ldots, M$.  

The above argument shows that if two minimal I-maps $\G_{\pi^{k}}$ and $\G_{\pi^{k-1}}$ along the given sequence are in the same MEC then their relative scores in Algorithm~1 are at most nondecreasing.  
Thus, it only remains to show that if $\G_{\pi^{k-1}}$ is not in the $\I$-MEC of $\G_{\pi^\ast}$ then
$$
\Score(\G_{\pi^\ast}) > \Score(\G_{\pi^{k-1}}).
$$
Since $\G_{\pi^{k-1}}$ and $\G_{\pi^\ast}$ are not $\I$-Markov equivalent then, by \cite[Theorem 10]{HB12}, there is at least one $I_t\in\I$ such that $\G_{\pi^{k-1}}^t$ and $\G_{\pi^\ast}^t$ have different skeletons.  
However, in this case $\Score(\G_{\pi^\ast}) > \Score(\G_{\pi^{k-1}})$ since the interventional distribution $\PP^t$ is faithful to the DAG $\G_{\pi^\ast}^t$.  \hfill$\square$

\bigskip

\section{Proof of Theorem~4.4}

We would now like to prove that Algorithm~2 is consistent under the faithfulness assumption.  
That is, suppose we are given a collection of interventional targets $\I = \{I_1 = \emptyset, I_2,\ldots, I_K\}$ and data drawn from the distributions $\PP^1,\ldots,\PP^K$, all of which are faithful to the (respective) interventional DAGs $\G_{\pi^\ast}^1,\ldots,\G_{\pi^\ast}^K$.  
Then Algorithm~2 will return a DAG that is $\I$-Markov equivalent to $\G_{\pi^\ast}$.
In \cite{SWUM17}, the authors show that there exists a sequence of I-MAPs given by covered arrow reversals  
$$
\G_\pi\geq \G_{\pi^1}\geq \cdots \geq \G_{\pi^{m-1}} \geq \G_{\pi^m}\geq  \cdots\geq \G_{\pi^M} \geq\G_{\pi^\ast}
$$
taking us from any $\G_\pi$ to the data-generating DAG $\G_{\pi^\ast}$.  
We must now show that there exists such a sequence using only $\I$-covered arrow reversals.   
\begin{theorem} 
\label{thm: I-covered arrow reversal sequence}
For any permutation $\pi$, there exists a list of $\I$-covered arrow reversals from $\G_\pi$ to the data-generating DAG $\G_{\pi^\ast}$
$$
\G_\pi = \G_{\pi^0}\geq \G_{\pi^1}\geq \cdots \geq \G_{\pi^{m-1}} \geq \G_{\pi^m}\geq  \cdots\geq \G_{\pi^{M-1}} \geq \G_{\pi^M}=\G_{\pi^\ast}
$$  
\end{theorem}

\begin{proof}
Suppose that $\G_{\pi^m}$ is produced from $\G_{\pi^{m-1}}$ via reversing the covered arrow $i\rightarrow j$ in $\G_{\pi^{m-1}}$ and let $S = \parents_{\G_{\pi^{m-1}}}(i)$.  
By Lemma~\ref{lem: d-connecting paths}, it must be that $i$ and $j$ are $d$-connected in $\G_{\pi^\ast}$ given $S$ only by paths for which the arrow incident to $i$ points towards $i$.  
It follows that for $k\in\I_{i\backslash j}$ there are no paths $d$-connecting $i$ and $j$ in $\G_{\pi^\ast}^k$.  
Therefore, $i\rightarrow j \notin \G_{\pi^{m-1}}$ for all $k\in\I_{i\backslash j}$; i.e., the arrow $i\rightarrow j$ is $\I$-covered in $\G_{\pi^{m-1}}$.  
\end{proof}

The previous theorem states that we can use only $\I$-covered arrow reversals to produce a sequence of I-MAPs taking us from any DAG $\G_\pi$ to the data-generating DAG $\G_{\pi^\ast}$.  
In the case that $\G_{\pi^{m-1}}$ and $\G_{\pi^m}$ are in different MECs it follows from the construction of such a sequence of minimal I-MAPs under the faithfulness assumption in the observational setting that $\G_{\pi^{m-1}}$ has strictly more arrows than $\G_{\pi^{m}}$.  
It remains to show that if $\G_{\pi^{m-1}}$ and $\G_{\pi^m}$ are in the true MEC then $\G_{\pi^{m-1}}$ has strictly more $\I$-contradicting arrows than $\G_{\pi^m}$ whenever they are not in the same $\I$-MEC and they have exactly the same $\I$-contradicting arrows when in the same $\I$-MEC.
This is the content of the following theorem.
\begin{theorem}
\label{thm: I-covered arrow reversals and different I-mecs}
Suppose that the distributions $\PP^1,\ldots,\PP^K$ are faithful to their respective interventional DAGs $\G_{\pi^\ast}^1,\ldots,\G_{\pi^\ast}^K$.  
For any permutation $\pi$ such that $\G_\pi$ and $\G_{\pi^\ast}$ are in the same MEC there exists a list of $\I$-covered arrow reversals from $\G_\pi$ to $\G_{\pi^\ast}$
$$
\G_\pi = \G_{\pi^0}\geq \G_{\pi^1}\geq \cdots \geq \G_{\pi^{m-1}} \geq \G_{\pi^m}\geq  \cdots\geq \G_{\pi^{M-1}} \geq \G_{\pi^M}=\G_{\pi^\ast}
$$  
such that, for all $m\in[M]$, either $\G_{\pi^{m-1}}$ and $\G_{\pi^{m}}$ are in the same $\I$-MEC or $\G_{\pi^{m}}$ is produced from $\G_{\pi^{m-1}}$ by the reversal of an $\I$-contradicting arrow.  
Moreover, the number of $\I$-contradicting arrows in $\G_{\pi^m}$ is strictly less than the number of $\I$-contradicting arrows in $\G_{\pi^{m-1}}$.  
\end{theorem}
\begin{proof}
Suppose that $\G_{\pi^{m}}$ is produced from $\G_{\pi^{m-1}}$ by reversing the $\I$-covered arrow $i\rightarrow j$ in $\G_{\pi^{m-1}}$.  
If $\I_{i\backslash j}  = \I_{j\backslash i} = \emptyset$ then $\G_{\pi^{m-1}}$ and $\G_{\pi^{m}}$ are in the same $\I$-MEC and hence $i \rightarrow j$ is not an $\I$-contradicting arrow.

Otherwise, they must belong to different $\I$-MECs and we have that $\I_{i\backslash j}  \cup \I_{j\backslash i} \neq \emptyset$. 
Let $S = \parents_{\G_{\pi^{m-1}}}(i)$, by Lemma~\ref{lem: d-connecting paths}.  
It must be that $i$ and $j$ are $d$-connected in $\G_{\pi^\ast}$ given $S$ only by paths for which the arrow incident to $i$ points towards $i$.  
Since $\G_{\pi^{m-1}}$ is in the true MEC then $i$ and $j$ must be adjacent in $\G_{\pi^\ast}$.  
It then follows from Lemma~\ref{lem: d-connecting paths} that the arrow between $i$ and $j$ points to $i$.  
In other words, $j \rightarrow i \in \G_{\pi^\ast}$. 
In this case, for all $k \in \I_{j \backslash i}$ we have, under the faithfulness assumption, that $\PP^k$ satisfies $i \not\independent j$ since the arrow $j \rightarrow i$ in $\G_{\pi^\ast}$ is not deleted in the interventional DAG $(\G_{\pi^\ast})^k$.  
Similarly, for all $k\in\I_{i\setminus j}$ we know $i\independent j$ in $\PP^k$ since all $d$-connecting paths between $i$ and $j$ in the interventional DAG $(\G_{\pi^\ast})^k$ must be given by conditioning on some descendants of $i$.  
Thus, $i$ and $j$ are $d$-separated given $\emptyset$ in $(\G_{\pi^\ast})^k$, and it follows from the Markov assumption that $i\independent j$ in $(\G_{\pi^\ast})^k$.
Therefore, by Definitions~4.2 and~4.3, we know that $i\rightarrow j$ is an $\I$-covered arrow that is also $\I$-contradicting.
Furthermore, since the reversal of an $\I$-contradicting arrow makes it not $\I$-contradicting and the $\I$-contradicting arrows of $\G_{\pi^m}$ are contained in the $\I$-contradicting arrows of $\G_{\pi^{m-1}}$, it follows that $\G_{\pi^m}$ has strictly less $\I$-contradicting arrows than $\G_{\pi^{m-1}}$.
\end{proof}

\emph{Proof of Theorem~4.4}  
The proof of this theorem follows immediately from Theorem~\ref{thm: I-covered arrow reversal sequence} and Theorem~\ref{thm: I-covered arrow reversals and different I-mecs}.  
This is because Theorem~\ref{thm: I-covered arrow reversal sequence} implies that under the faithfulness assumption there is a sequence of $\I$-covered arrow reversals by which we can reach the true MEC, and Theorem~\ref{thm: I-covered arrow reversals and different I-mecs} implies that we then use $\I$-contradicting arrows to reach the true $\I$-MEC within the true MEC.  
Moreover, Theorem~\ref{thm: I-covered arrow reversals and different I-mecs} implies that the true $\I$-MEC will contain DAGs with the fewest $\I$-contradicting arrows.  
    \hfill$\square$

\section{Supplementary Material for Real Data Analysis}

This section contains supplementary information about the real data analysis. 
Table~\ref{table:cells_per_int} and Figure~\ref{fig:original_heatmap} present additional details about the perturb-seq experiments. 
Table~\ref{table:cells_per_int_protein} shows more details about the flow cytometry dataset. 
Table~\ref{table:edge_list} compares the results of IGSP and k-IGSP with other methods that allow latent confounders as applied to the Sachs et al.~\cite{Sachs_2005} dataset.
Figures~\ref{fig: perturb seq network} and~\ref{fig: protein signalling network} are our reconstructions of the causal gene network for the perturb-seq data set and the protein signaling network for the Sachs et al. dataset, respectively. 


\begin{figure}[h!]
\centering
\includegraphics[width=0.9\textwidth]{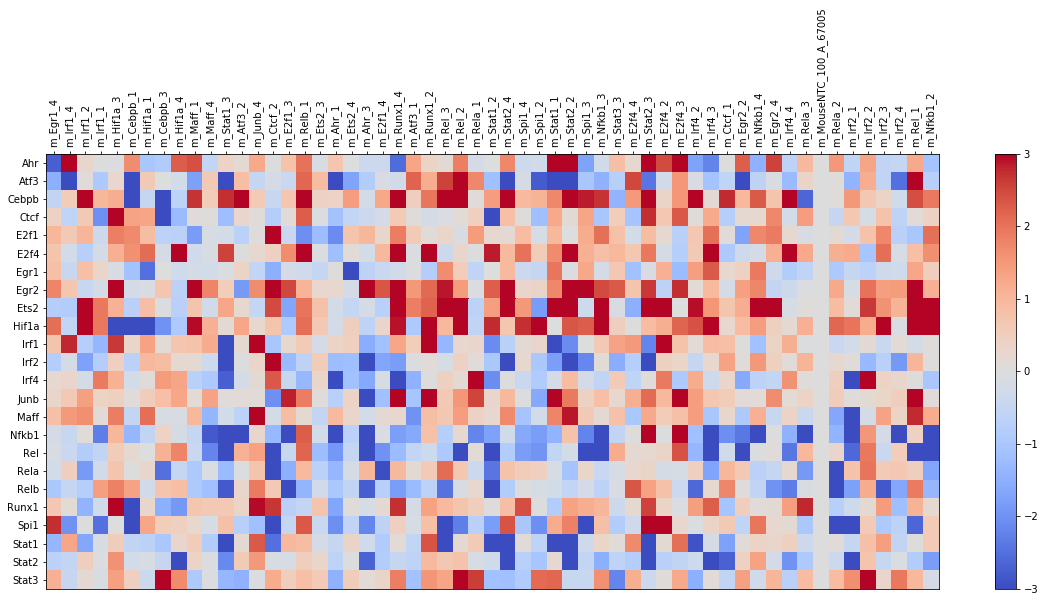}
\caption{
Heatmap of the true effects of each gene deletion on each measured gene. All 56 guide RNAs used in the experiment are listed on the x-axis and measured genes are listed on the y-axis. 18 of 56 guides, which target 8 genes in total, were selected for analysis because they were effective. Red (positive on q-value scale) indicate gene deletions that increase abundance of the measured gene. Blue (negative on q-value scale) indicate gene deletions that decrease abundance of the measured gene. White (zero on q-value scale) indicates no observed effect of gene deletion.
}
\label{fig:original_heatmap}
\end{figure}

\begin{figure}[h!]
\centering
\subfigure[Ground truth]{\includegraphics[width=0.45\textwidth]{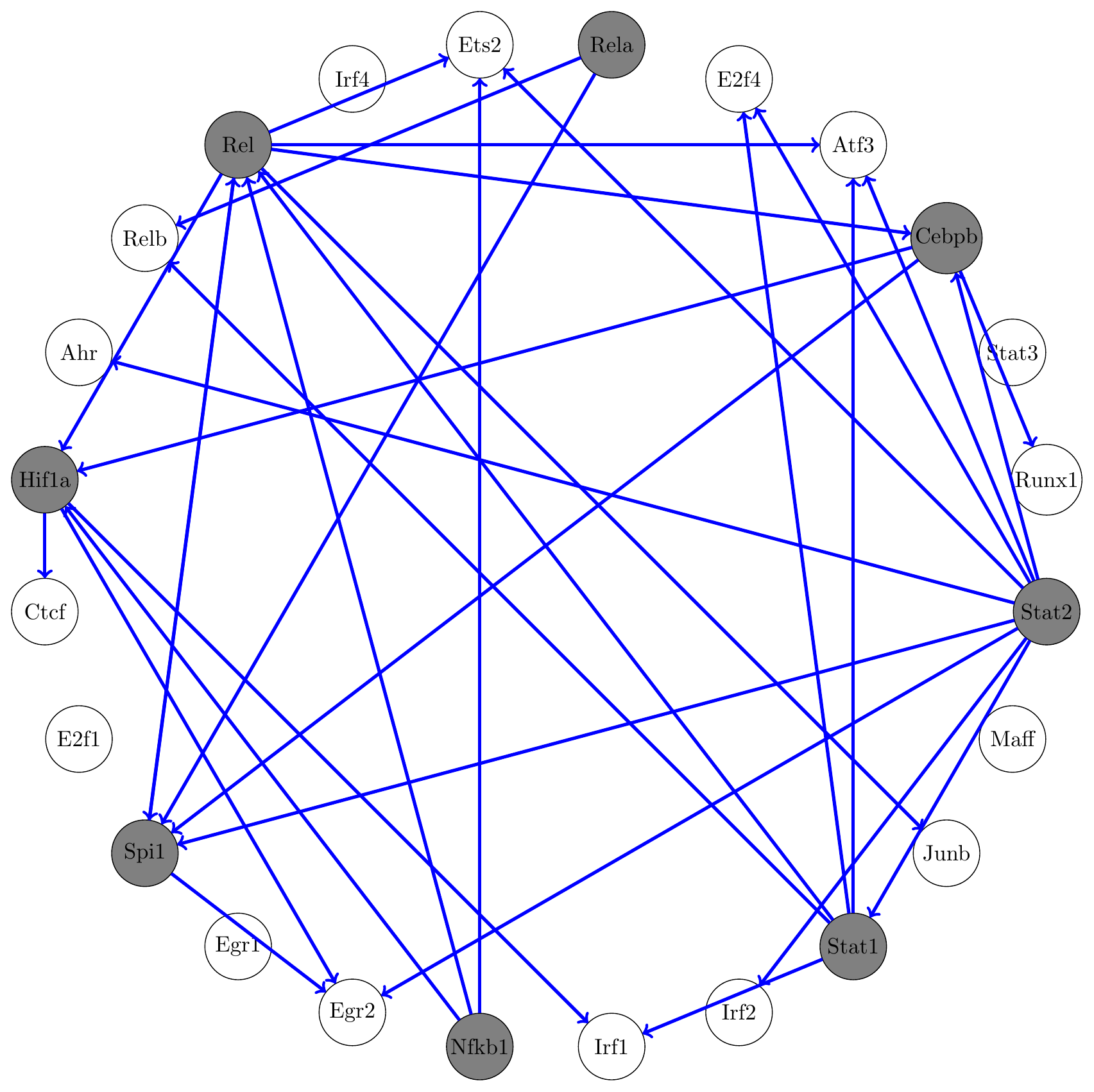}\label{fig:perturb_net1}}
\subfigure[Gaussian CI test]{\includegraphics[width=0.45\textwidth]{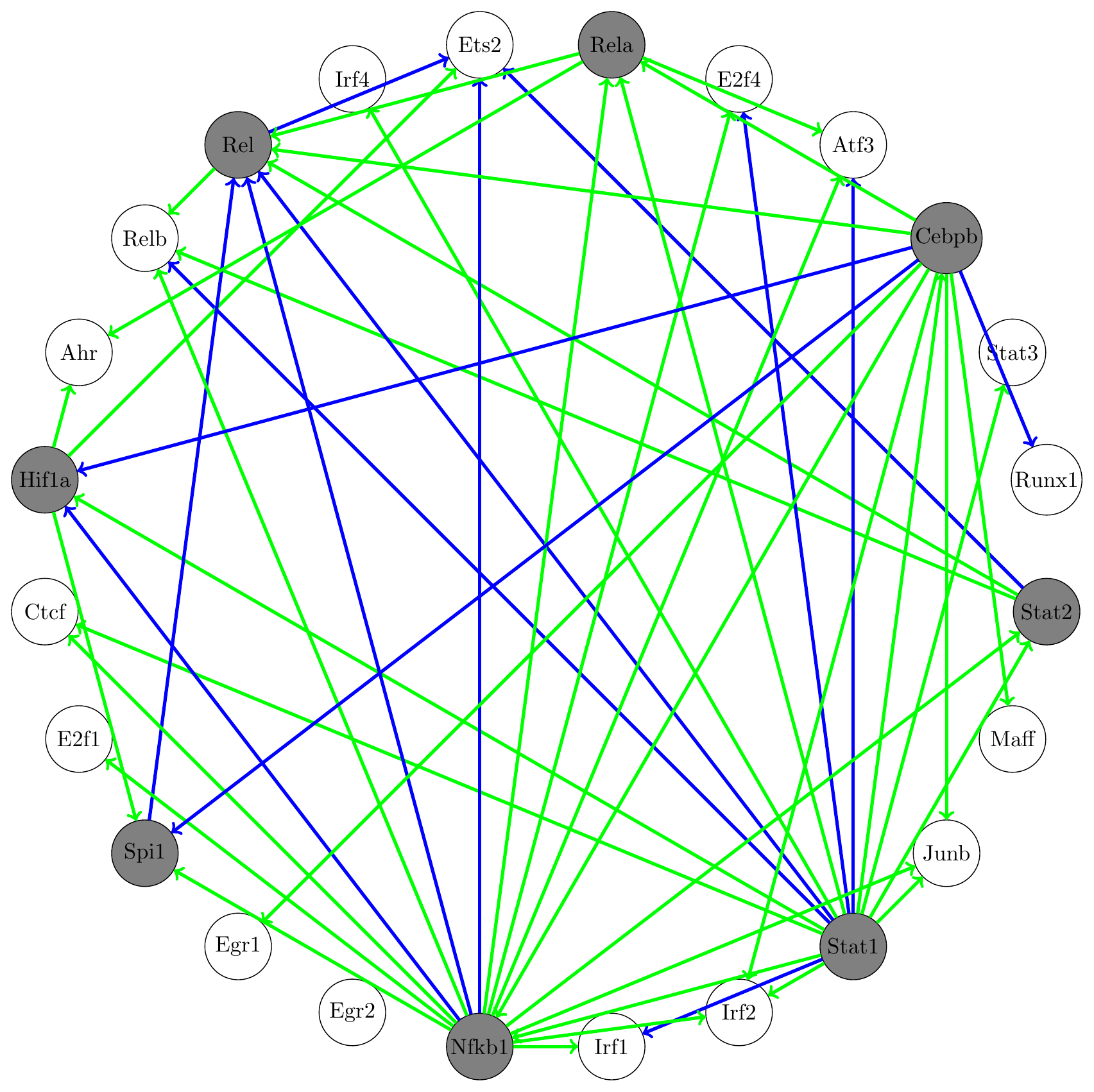}\label{fig:perturb_net2}}
\caption{Partial causal gene networks from perturb-seq data. Gray nodes represent genes receiving interventions (a) Ground truth partial causal network obtained from perturb-seq data, i.e. from thresholding the values in the heatmap from Figure 5(a). (b) Reconstruction of partial causal gene network using Algorithm~2 with Gaussian CI test (cut-off value $\alpha = 0.15$). Here, blue edges are the true positive edges, and green edges are the false positive edges.
}
	\label{fig: perturb seq network}
\end{figure}

\begin{table}[h!]
\centering
\caption{Number of samples under each gene deletion for processed perturb-seq dataset}
\label{table:cells_per_int}
\begin{tabular}{||c c c c c c c c c c||} 
 \hline
 Intervention: & None & Stat2 & Stat1 & Rel & Hif1a & Spi1 & Nfkb1 & Rela & Cebpb \\ [0.5ex] 
 \hline
 \# Samples: & 992 & 2426 & 3337 & 1513 & 301 & 796 & 3602 & 1068 & 392 \\ [1ex] 
 \hline
\end{tabular}
\end{table}

\vspace{0.12in}

\begin{table}[!h]
\centering
\caption{Number of samples under each protein intervention for flow cytometry dataset}
\label{table:cells_per_int_protein}
\begin{tabular}{||c c c c c c c||} 
 \hline
 Intervention: & None & Akt & PKC & PIP2 & Mek & PIP3 \\ [0.5ex] 
 \hline
 \# Samples: & 1755 & 911 & 723 & 810 & 799 & 848 \\ [1ex] 
 \hline
\end{tabular}
\end{table}

\vspace{0.12in}

\begin{table}[!h]
\centering
\caption{Interaction prediction results of IGSP and k-IGSP and other methods that allow for latent variables. Here the consensus network from~\cite{Sachs_2005} is denoted by~\cite{Sachs_2005}a and their reconstructed network by~\cite{Sachs_2005}b. For~\cite{Meinshausen} we provide results from both ICP and hidden ICP, denoted by~\cite{Meinshausen}a and~\cite{Meinshausen}b respectively. For IGSP and k-IGSP, we chose the standardly used significance level $\alpha=0.05$ as the cut-off for CI testing, which resulted in a similar number of predicted interactions as in~\cite{Meinshausen}.}
\label{table:edge_list}
\begin{tabular}{c| c c c c c c c c} 
\hline
Edge & ~\cite{Sachs_2005}a & ~\cite{Sachs_2005}b & ~\cite{Meinshausen}a & ~\cite{Meinshausen}b & ACI~\cite{MCM16} & COmbINE~\cite{TT15} & IGSP & k-IGSP \\
\hline
RAF $\to$ MEK  &  \cc  &  \cc  &   &  \cc  &   & \cc &  & \\
RAF $\to$ JNK &  &  &  &  &  &  &  &  \cc \\
MEK $\to$ RAF  &   &  &   &  \cc  &  \cc  & \cc &  \cc  &  \cc \\
MEK $\to$ ERK  &  \cc  &  \cc  &   &  &  \cc  &  &  & \\
MEK $\to$ AKT  &   &  &   &  &  \cc  &  &  & \\
MEK $\to$ JNK  &   &  &   &  &  \cc  &  &  & \\
PLCg $\to$ PIP2  &  \cc  &  \cc  &  \cc  &  \cc  &   & \cc &  & \\
PLCg $\to$ PIP3  &   &  \cc  &   &  &   & \cc &  \cc  & \\
PLCg $\to$ PKC  &  \cc  &  &   &  &   &  &  & \\
PIP2 $\to$ PLCg  &   &  &  \cc  &  &  \cc  & \cc &  &  \cc \\
PIP2 $\to$ PIP3  &   &  &   &  &   &  &  &  \cc \\
PIP2 $\to$ PKC  &  \cc  &  &   &  &   &  &  & \\
PIP3 $\to$ PLCg  &  \cc  &  &   &  &   & \cc &  &  \cc \\
PIP3 $\to$ PIP2  &  \cc  &  \cc  &  \cc  &  \cc  &  & \cc &  \cc  & \\
PIP3 $\to$ AKT  &  \cc  &  &   &  &  &  &  & \\
AKT $\to$ ERK  &   &  &  \cc  &  \cc  &  & \cc &  \cc  &  \cc \\
AKT $\to$ PKA &  &  &  &  &  & \cc &  \cc  & \\
AKT $\to$ JNK  &   &  &   &  &  \cc &  &  & \\
ERK $\to$ AKT  &   &  \cc  &  \cc  &  \cc  &  & \cc &  & \\
PKA $\to$ RAF  &  \cc  &  \cc  &   &  &  \cc  &  &  & \\
PKA $\to$ MEK  &  \cc  &  \cc  &   &  \cc  &  \cc &  &  & \\
PKA $\to$ ERK  &  \cc  &  \cc  &  \cc  &  &  \cc & \cc &  &  \cc \\
PKA $\to$ AKT  &  \cc  &  \cc  &   &  \cc  &  \cc &  &  &  \cc \\
PKA $\to$ PKC  &   &  &   &  &   &  &  & \\
PKA $\to$ P38  &  \cc  &  \cc  &   &  &  \cc &  &  & \\
PKA $\to$ JNK  &  \cc  &  \cc  &   &  &  \cc &  &  & \\
PKC $\to$ RAF  &  \cc  &  \cc  &   &  &  \cc &  &  & \\
PKC $\to$ MEK  &  \cc  &  \cc  &   &  &  \cc &  &  & \\
PKC $\to$ PLCg  &   &  &   &  &  \cc &  &  & \\
PKC $\to$ PIP2  &   &  &   &  &  \cc &  &  & \\
PKC $\to$ PIP3  &   &  &   &  &  \cc &  &  & \\
PKC $\to$ ERK   &   &  &   &  &  \cc &  &  & \\
PKC $\to$ AKT  &   &  &   &  &  \cc &  &  & \\
PKC $\to$ PKA  &   &  \cc  &   &  &  &  &  & \\
PKC $\to$ P38  &  \cc  &  \cc  &   &  \cc  &  \cc &  &  \cc  &  \cc \\
PKC $\to$ JNK  &  \cc  &  \cc  &  \cc  &  \cc  &  \cc &  &  \cc  &  \cc \\
P38 $\to$ JNK  &   &  &   &  \cc  &  & \cc &  &  \cc \\
P38 $\to$ PKC  &   &  &   &  \cc  &  &  &  & \\
JNK $\to$ PKC  &   &  &   &  \cc  &  &  &  & \\
JNK $\to$ P38  &   &  &   &  \cc  &  & \cc &  & \\
\hline 
\end{tabular}
\end{table}

\begin{figure}[h!]
\centering
\subfigure[Ground truth]{\includegraphics[width=0.32\textwidth]{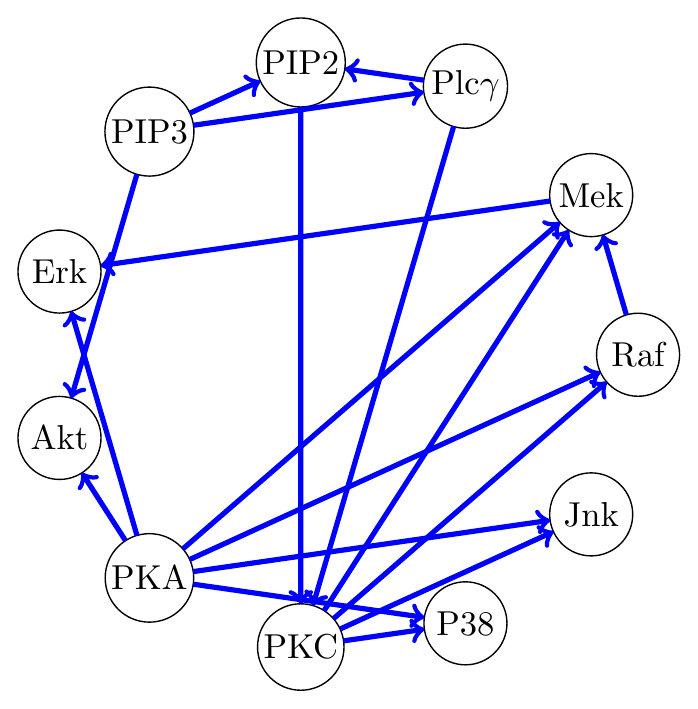}\label{fig:net1}}
\subfigure[Gaussian CI test]{\includegraphics[width=0.32\textwidth]{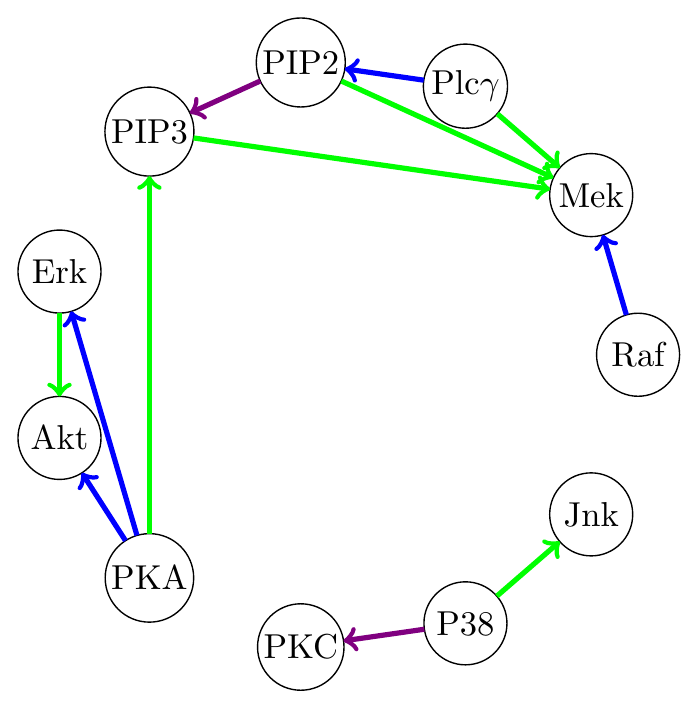}\label{fig:net2}}
\subfigure[Kernel CI test]{\includegraphics[width=0.32\textwidth]{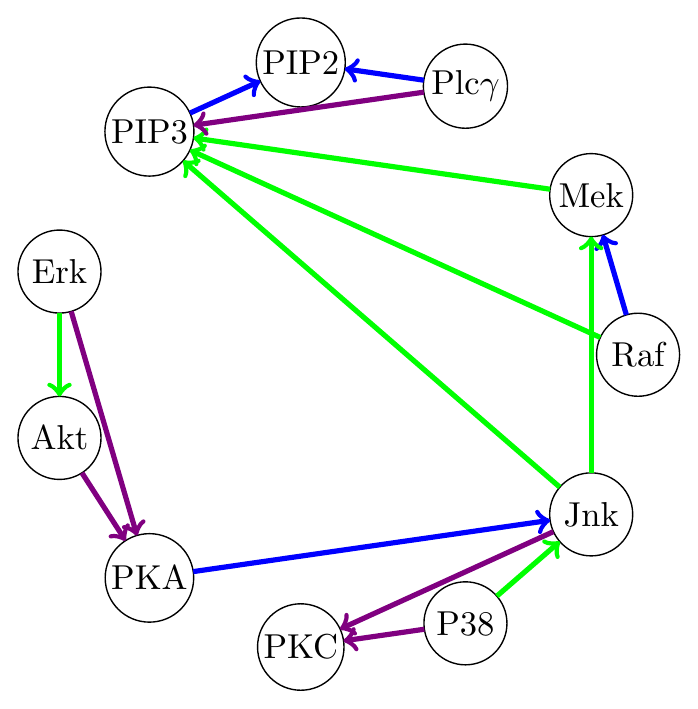}\label{fig:net3}}
\caption{Protein signaling network from flow cytometry data. (a) Ground truth network according to the conventionally accepted model from~\cite{Sachs_2005}. (b) Reconstruction of protein signaling network using Algorithm~2 with Gaussian CI test (cut-off value $\alpha = 0.005$). Here blue edges are the true positive edges; purple edges are the reversed edges that share the same skeleton as ground truth edges but the arrows are different; green edges are the false positives. (c) Reconstruction of protein signaling network using Algorithm~2 with kernel-based CI test (cut-off value $\alpha = 0.0001$). Here we choose different significance level cut-offs for kernel-based test and Gaussian test such that the number of true positive directed edges are the same.
}
\vspace*{5in}
	\label{fig: protein signalling network}
\end{figure}

\end{document}